\documentclass[runningheads,a4paper,orivec]{llncs2e/llncs}

\usepackage[normalweight=Book]{fdsymbol}

\usepackage{blindtext}
\usepackage{url}
\usepackage{color}
\usepackage{nicefrac}
\usepackage{graphicx}
\usepackage{multirow}
\usepackage{rotate}
\usepackage{stackengine} 
\usepackage{bm} 
\usepackage{enumerate}
\usepackage{thmtools}
\usepackage{thm-restate}
\usepackage{relsize} 
\usepackage[noend]{algorithm2e} 
\usepackage{multicol}
\usepackage{wrapfig}
\usepackage[bookmarks=false,linktocpage=true,colorlinks=false]{hyperref} 

\newif\ifcommentson\commentsonfalse
\def\mywidth{.9}
\def\mywidthRep{.8}
\ifcommentson
\newcommand{\commentCP}[1]{\begin{center} \parbox{\mywidth\textwidth}{\textbf{\textcolor{black}{Comment C.}} \textcolor{red}{#1 }}\end{center}}
\newcommand{\commentKC}[1]{\begin{center} \parbox{\mywidth\textwidth}{\textbf{\textcolor{black}{Comment K.}} \textcolor{red}{#1} }\end{center}}
\newcommand{\commentYK}[1]{\begin{center} \parbox{\mywidth\textwidth}{\textbf{\textcolor{black}{Comment Y.}} \textcolor{red}{#1} }\end{center}}
\newcommand{\commentMA}[1]{\begin{center} \parbox{\mywidth\textwidth}{\textbf{\textcolor{black}{Comment M.}} \textcolor{red}{#1} }\end{center}}
\newcommand{\replyCP}[1]{\begin{center} \parbox{\mywidthRep\textwidth}{\textbf{Reply C.} \textcolor{blue}{#1} }\end{center}}
\newcommand{\replyKC}[1]{\begin{center} \parbox{\mywidthRep\textwidth}{\textbf{Reply K.} \textcolor{blue}{#1} }\end{center}}
\newcommand{\replyYK}[1]{\begin{center} \parbox{\mywidthRep\textwidth}{\textbf{Reply Y.} \textcolor{blue}{#1} }\end{center}}
\newcommand{\replyMA}[1]{\begin{center} \parbox{\mywidthRep\textwidth}{\textbf{Reply M.} \textcolor{blue}{#1} }\end{center}}
\newcommand{\commentC}[1]{\marginpar{\footnotesize \color{red} {\bf C:} \textsf{\scriptsize #1}}}
\newcommand{\commentK}[1]{\marginpar{\footnotesize \color{red} {\bf K:} \textsf{\scriptsize #1}}}
\newcommand{\commentY}[1]{\marginpar{\footnotesize \color{red} {\bf Y:} \textsf{\scriptsize #1}}}
\newcommand{\commentM}[1]{\marginpar{\footnotesize \color{red} {\bf M:} \textsf{\scriptsize #1}}}
\newcommand{\replyC}[1]{\marginpar{\footnotesize \color{red} {\bf C:} \textsf{\scriptsize #1}}}
\newcommand{\replyK}[1]{\marginpar{\footnotesize \color{red} {\bf K:} \textsf{\scriptsize #1}}}
\newcommand{\replyY}[1]{\marginpar{\footnotesize \color{red} {\bf Y:} \textsf{\scriptsize #1}}}
\newcommand{\replyM}[1]{\marginpar{\footnotesize \color{red} {\bf M:} \textsf{\scriptsize #1}}}

\else
\newcommand{\commentCP}[1]{}
\newcommand{\commentKC}[1]{}
\newcommand{\commentYK}[1]{}
\newcommand{\commentMA}[1]{}
\newcommand{\replyCP}[1]{}
\newcommand{\replyKC}[1]{}
\newcommand{\replyYK}[1]{}
\newcommand{\replyMA}[1]{}
\newcommand{\commentC}[1]{}
\newcommand{\commentK}[1]{}
\newcommand{\commentY}[1]{}
\newcommand{\commentM}[1]{}
\newcommand{\replyC}[1]{}
\newcommand{\replyK}[1]{}
\newcommand{\replyY}[1]{}
\newcommand{\replyM}[1]{}

\fi

\newtheorem{Theorem}{Theorem}

\newtheorem{Definition}[Theorem]{Definition}
\newtheorem{Example}[Theorem]{Example}

\newcommand{\calx}{\mathcal{X}}
\newcommand{\caly}{\mathcal{Y}}

\newcommand{\cala}{\mathcal{A}}
\newcommand{\cald}{\mathcal{D}}

\newcommand{\calad}{\mathcal{A}\,{\rightarrow}\,\mathcal{D}}
\newcommand{\calda}{\mathcal{D}\,{\rightarrow}\,\mathcal{A}}
\newcommand{\cali}{\mathcal{I}}
\newcommand{\calj}{\mathcal{J}}
\newcommand{\reals}{\mathbb{R}}
\newcommand{\distr}{\mathbb{D}}

\newcommand{\infoset}{\mathit{K}}

\newcommand{\infoseta}{\infoset_{\mathsf{a}}}

\newcommand{\psd}{\mathit{s_{\sf d}}}
\newcommand{\psa}{\mathit{s_{\sf a}}}
\newcommand{\msd}{\mathit{\sigma_{\sf d}}}
\newcommand{\msa}{\mathit{\sigma_{\sf a}}}
\newcommand{\payd}{\mathit{u_{\sf d}}}
\newcommand{\paya}{\mathit{u_{\sf a}}}
\newcommand{\pay}{\mathit{u}}
\newcommand{\Payd}{\mathit{U_{\sf d}}}
\newcommand{\Paya}{\mathit{U_{\sf a}}}
\newcommand{\Pay}{\mathit{U}}

\newcommand{\eqdef}{\ensuremath{\stackrel{\mathrm{def}}{=}}}
\newcommand{\argmax}{\operatornamewithlimits{argmax}}
\newcommand{\argmin}{\operatornamewithlimits{argmin}}
\newcommand{\supp}[1]{{\sf supp}(#1)}

\newcommand{\expectDouble}[2]{\operatornamewithlimits{\displaystyle\mathbb{E}}_{\substack{#1\\ #2}}}

\renewcommand{\equiv}{\approx}

\newcommand{\vg}{V_{g}} 
\newcommand{\priorvg}[1]{\vg\left[#1\right]} 
\newcommand{\vf}{\mathbb{V}} 
\newcommand{\priorvf}[1]{\vf\left[#1\right]} 
\newcommand{\postvf}[2]{\vf\left[#1,#2\right]} 

\newcommand{\samplefrom}{\leftarrow} 

\newcommand{\add}{\operatorname{+}}
\newcommand{\bigadd}{\operatorname{\sum}}
\newcommand{\hchoice}[1]{\;{{}_{\mathit{#1}}{\mathlarger{\mathlarger{\oplus}}}}\;}
\newcommand{\HChoice}[2]{{\osum_{\mathit{#1} \samplefrom \mathit{#2}}}}
\newcommand{\HChoiceDouble}[4]{{\osum_{\substack{\mathit{#1} \samplefrom \mathit{#2}\\ \mathit{#3} \samplefrom \mathit{#4}}}}}
\newcommand{\hchoiceop}{\osum} 

\newcommand{\conc}{\diamond} 
\newcommand{\bigconc}{\meddiamond}
\newcommand{\vchoice}[1]{\;{{}_{\mathit{#1}}{\mathlarger{\mathlarger{\sqcupdot}}}}\;}
\newcommand{\VChoice}[2]{{\bigsqcupdot_{\mathit{#1} \samplefrom \mathit{#2}}}}
\newcommand{\VChoiceDouble}[4]{{\bigsqcupdot_{\substack{\mathit{#1} \samplefrom \mathit{#2}\\ \mathit{#3} \samplefrom \mathit{#4}}}}}
\newcommand{\vchoiceop}{\bigsqcupdot} 

\newcommand{\dist}[1]{\mathbb{D}{#1}}

\newcommand{\qm}[1]{``#1''}
\newcommand{\true}{T}
\newcommand{\false}{F}

\usepackage{subscript}

\title{Leakage and Protocol Composition in a Game-Theoretic Perspective\thanks{This version corresponds to the paper published in the POST 2018 conference proceedings. We suggest the reader to read instead the full version of this paper, available at~\cite{Alvim18:Entropy:arXiv}, also because in this version there are some minor mistakes, which have been corrected in the full version.}}
\titlerunning{Leakage and Protocol Composition in a Game-Theoretic Perspective}
\author{M\'{a}rio S. Alvim\inst{1}
\and Konstantinos Chatzikokolakis\inst{2}
\and Yusuke Kawamoto\inst{3} 
\and Catuscia Palamidessi\inst{4}}
\authorrunning{Alvim et al.}

\institute{Universidade Federal de Minas Gerais, Belo Horizonte, Brazil
\and CNRS and \'{E}cole Polytechnique, Palaiseau, France
\and AIST, Tsukuba, Japan
\and INRIA and \'{E}cole Polytechnique, Palaiseau, France}

\toctitle{Lecture Notes in Computer Science}
\tocauthor{Authors' Instructions}

\begin{document}

\maketitle

\begin{abstract}
In the inference attacks studied in Quantitative Information Flow (QIF),
the adversary typically tries 
to interfere with the system in the attempt 
to increase its leakage of secret information. 
The defender, on the other hand, 
typically tries 
to decrease leakage by introducing some controlled noise.
This noise introduction can be modeled as a type of protocol composition, i.e., 
a probabilistic choice among different protocols, and its effect on
the amount of leakage depends heavily on whether or not this choice is 
visible to the adversary.
In this work we consider operators for modeling visible and
invisible choice in protocol composition, and we study 
their algebraic properties.
We then formalize the interplay between defender and adversary in a 
game-theoretic framework adapted to the specific issues of QIF, 
where the payoff is information leakage. 
We consider various kinds of leakage games, depending on whether players 
act simultaneously or sequentially, and on whether or not the choices of the 
defender are visible to the adversary.
Finally, we establish a hierarchy of these games in terms
of their information leakage, and provide methods for finding 
optimal strategies (at the points of equilibrium) for both attacker and
defender in the various cases.
The full version of this paper can be found in \url{https://arxiv.org/abs/1803.10042}
\end{abstract}

\section{Introduction}
\label{sec:introdction}
A fundamental problem in computer security is the leakage of sensitive information due to  
\emph{correlation} of secret values with 
\emph{observables}---i.e., any information accessible to the attacker, 
such as, for instance, the system's outputs or execution time.
The typical defense consists in reducing this correlation, 
which can be done in, essentially, two ways.
The first, applicable when the correspondence secret-observable is deterministic, 
consists in coarsening the equivalence classes of 
secrets that give rise to the same observables. 
This can be 
achieved with post-processing, i.e., sequentially composing the original system 
with a program that removes information from observables.
For example, a typical attack on encrypted web traffic consists on the analysis of the packets' 
length, and a typical defense consists in padding extra bits so to diminish the length variety 
\cite{sun:02:SandP}. 

The second kind of defense, on which we focus in this work, 
consists in adding controlled noise to the observables produced by
the system. 
This can be usually seen as a composition of different protocols via 
probabilistic choice.
\begin{example}[Differential privacy]
Consider a counting query $f$, namely a function that, applied to a dataset $x$, returns the number of individuals in $x$ that satisfy a given property. A way to implement differential privacy \cite{Dwork:06:TCC} is to add geometrical noise to the result of $f$, so to obtain a probability distribution $P$ on integers of the form $P(z) = c \, e^{|z-f(x)|}$, where $c$ is a normalization factor. The resulting mechanism can  be interpreted as a probabilistic choice on protocols of the form $f(x), f(x) + 1, f(x) + 2, \ldots , f(x) - 1, f(x)-2, \ldots$, where the probability assigned to 
 $f(x) + n$ and to $f(x)-n$ decreases exponentially with $n$.
\end{example}
\begin{example}[Dining cryptographers]
Consider two agents running the dining cryptographers protocol \cite{Chaum:88:JC}, which consists in tossing a fair binary coin and then declaring 
the exclusive or $\oplus$ of their secret value $x$ and the result of the coin.  The protocol can be thought as the fair probabilistic choice  of two protocols, one consisting simply of declaring $x$, and the other declaring $x\oplus 1$. 
\end{example}

Most of the work in the literature of quantitative information flow (QIF) considers passive attacks,
in which the adversary only observes the system. 
Notable exceptions are the works \cite{Boreale:15:LMCS,Mardziel:14:SP,Alvim:17:GameSec}, 
which consider attackers 
who interact with and influence the system, possibly in an adaptive way, with
the purpose of maximizing the leakage of information. 

\begin{example}[CRIME attack]\label{exe:Crime}
Compression Ratio Info-leak Made Easy (CRIME) \cite{Rizzo:12:Ekoparty} is a
security exploit against secret web cookies over connections using the HTTPS and SPDY protocols and  data compression. 
The idea is that the attacker can inject some content $a$ in the communication of the secret $x$ from the target site to the server. The server then compresses and  encrypts the data, including both $a$ and $x$, and sends back the result. By observing the length of the result, the attacker can then infer information about $x$. 
To mitigate the leakage, one possible defense would consist in transmitting, along with $x$, also an encryption method $f$ selected randomly from a set $F$. Again, the resulting protocol can be seen as a composition, using probabilistic choice, of the protocols in the set $F$.  
\end{example}

In all examples above the main use of the probabilistic choice is to  
obfuscate the relation between secrets and observables, 
thus reducing their correlation---and, hence, the information leakage. 
To achieve this goal, it is essential that the attacker never comes to 
know the result of the choice. 
In the CRIME example, however, if 
$f$ and  $a$ are chosen independently, then (in general) it is still better 
to choose $f$ probabilistically, even if the adversary will come to know, afterwards, 
the choice of $f$. 
In fact, this is true also for the attacker: his best strategies (in general) are to chose $a$
according to some probability distribution. 
Indeed, suppose that $F=\{f_1,f_2\}$ are the defender's choices and 
$A=\{a_1,a_2\}$ are the attacker's, and that 
$f_1(\cdot,a_1)$ leaks more than $f_1(\cdot,a_2)$, while $f_2(\cdot,a_1)$ 
leaks less than $f_2(\cdot,a_2)$. 
This is a scenario like \emph{the matching pennies} in game theory: 
if one player selects an action deterministically, the other player 
may exploit this choice and  get an advantage. 
For each player the optimal strategy is to play probabilistically, 
using a distribution that maximizes his own gain for all possible actions 
of the adversary. 
In zero-sum games, in which the gain of one player coincides 
with the loss of the other, the optimal pair of distributions always exists, 
and it is called \emph{saddle point}. 
It also coincides with the \emph{Nash equilibrium}, which is defined as the 
point in which neither of the two players gets any advantage in changing 
unilaterally his strategy. 

Motivated by these examples, this paper investigates the two kinds of choice, visible 
and hidden (to the attacker), in a game-theoretic setting. 
Looking at them as language operators, we study their algebraic properties, which will
help reason about their behavior in games. 
We consider zero-sum games, in which the gain (for the attacker) is represented by the leakage. 
While for visible choice it is appropriate to use the ``classic'' game-theoretic
framework, for hidden  choice we need to adopt the more general framework of the
\emph{information leakage games} proposed in \cite{Alvim:17:GameSec}. 
This happens because, in contrast with standard game theory, 
in games with hidden choice the utility of a mixed strategy is 
a convex function of the distribution on the defender's pure actions, 
rather than simply the expected value of their utilities.
We will consider both simultaneous games---in which each player chooses independently---and sequential games--- in which one player chooses his action first. We aim at comparing all these situations, and at identifying the precise advantage of the hidden choice over the visible one. 

To measure leakage we use the well-known information-theoretic model. 
A central notion in this model is that of \emph{entropy}, but here
we  use its converse, \emph{vulnerability}, which 
represents  the magnitude of the threat. 
In order to derive results as general as possible, we  adopt 
the very comprehensive notion of vulnerability as any convex and continuous 
function, as  used in \cite{Boreale:15:LMCS} and 
\cite{Alvim:16:CSF}. 
This notion has been shown \cite{Alvim:16:CSF} to subsume most 
 information measures, including 
\emph{Bayes vulnerability} (aka min-vulnerability, aka (the converse of) 
Bayes risk)~\cite{Smith:09:FOSSACS,Chatzikokolakis:08:JCS}, 
\emph{Shannon entropy}~\cite{Shannon:48:Bell}, 
\emph{guessing entropy}~\cite{Massey:94:IT}, and 
\emph{$g$-vulnerability}~\cite{Alvim:12:CSF}.\\

The main contributions of this paper are:
\begin{itemize}
\item We present a general framework for reasoning about information leakage
in a game-theoretic setting, extending the notion of information leakage games proposed 
in~\cite{Alvim:17:GameSec} to both simultaneous and sequential games,
with either hidden or visible choice.

\item We present a rigorous compositional way, using visible and hidden choice operators, for representing adversary and defender's actions in information leakage games.
In particular, we study the algebraic properties of visible and hidden choice on channels,
and compare the two kinds of choice with respect to the capability of reducing 
leakage, in  presence of an adaptive attacker.

\item 
We provide a taxonomy of the various scenarios (simultaneous and sequential) 
showing when randomization is necessary, for either attacker or defender, 
to achieve optimality.
Although it is well-known in information flow that the defender's best strategy 
is usually randomized, only recently it has been shown that
when defender and adversary act simultaneously, the adversary's optimal strategy 
also requires randomization~\cite{Alvim:17:GameSec}.

\item We use our framework in a detailed case study of a password-checking protocol. 
The naive program, which checks the password bit by bit and stops when it finds a 
mismatch, is clearly very insecure, because it reveals at each attempt the maximum 
correct prefix. 
On the other hand,  if we continue checking until the end of the string (time padding), 
the program becomes very inefficient. 
We show that, by using probabilistic choice instead, 
we can obtain a good trade-off between security and efficiency.
\end{itemize}

\paragraph*{Plan of the paper.}
The remaining of the paper is organized as follows.
In Section~\ref{sec:preliminaries} we review some basic notions of game theory and 
quantitative information flow. 
In Section~\ref{sec:running-example} we introduce our running example. 
In Section~\ref{sec:operators} we define the visible and hidden choice operators 
and demonstrate their algebraic properties. 
In Section~\ref{sec:games-setup}, the core of the paper, we examine
various scenarios for leakage games. 
In Section~\ref{sec:password-example} we show an application of our framework 
to a password checker. 
In Section~\ref{sec:related-work} we discuss related work and,
finally, in Section~\ref{sec:conclusion} we conclude.   

\section{Preliminaries}
\label{sec:preliminaries}
In this section we review some basic notions from game theory
and quantitative information flow.
We use the following notation:
Given a set $\cali$, we denote by $\distr\cali$ the 
\emph{set of all probability distributions} over $\cali$.
Given $\mu\in \distr\cali$, its \emph{support}
$\supp{\mu} \eqdef \{ i \in \cali : \mu(i)>0 \}$ 
is the set of its elements with positive probability.
We use $i{\samplefrom}\mu$ to indicate that a value 
$i{\in}\cali$ is sampled from a distribution $\mu$ on $\cali$.

\subsection{Basic concepts from game theory}
\label{subsec:game-theory}

\subsubsection{Two-player games}

\emph{Two-player games} are a model for reasoning about 
the behavior of two players.
In a game, each player has at its disposal a set of \emph{actions} 
that he can perform, and he obtains some gain or loss depending on
the actions chosen by both  players.
Gains and  losses are defined using a real-valued \emph{payoff
function}.
Each player is assumed to be \emph{rational}, i.e.,  
his choice is driven by the attempt to maximize his own expected payoff.
We also assume that 
the set of possible actions and the payoff functions of both players 
are \emph{common knowledge}.

In this paper we only consider \emph{finite games}, in which the 
set of actions available to the players are finite.   
Next we introduce an important distinction between \emph{simultaneous} 
and \emph{sequential} games. 
In the following, we will call the two players \emph{defender} 
and \emph{attacker}.

\subsubsection{Simultaneous games}

In a  simultaneous game, each player chooses his action without knowing 
the action chosen by the other.
The term ``simultaneous'' here does  not  mean that the players' actions 
are chosen at the same time, but only that they are chosen independently.
Formally, such a game  is defined as a tuple\footnote{Following the convention of \emph{security games}, we set the first player to be the defender.} 
$(\cald,\, \cala,\, \payd, \paya)$, where
$\cald$ is a nonempty set of \emph{defender's actions}, 
$\cala$ is a nonempty set of \emph{attacker's actions},
$\payd: \cald\times\cala \rightarrow \reals$ is the \emph{defender's payoff function}, and
$\paya: \cald\times\cala \rightarrow \reals$ is the \emph{attacker's payoff function}.

Each player may choose an action deterministically or probabilistically.
A \emph{pure strategy} of the defender (resp. attacker) is a deterministic 
choice of an action, i.e., an element $d\in\cald$ (resp. $a\in\cala$). 
A pair $(d, a)$ is called \emph{pure strategy profile}, and 
$\payd(d, a)$, $\paya(d, a)$ represent the defender's and the attacker's 
payoffs, respectively. 
A \emph{mixed strategy} of the defender (resp. attacker) is a probabilistic 
choice of an action, defined as a probability distribution 
$\delta\in\distr\cald$ (resp. $\alpha\in\distr\cala$).
A pair  $(\delta, \alpha)$ is called  \emph{mixed strategy profile}.
The defender's and the attacker's \emph{expected payoff functions} for
mixed strategies are defined, respectively, as:
$
\Payd(\delta,\alpha)
\eqdef {\expectDouble{d\leftarrow\delta}{a\leftarrow\alpha} \payd(d, a)}
=\sum_{\substack{d\in\cald\\ a\in\cala}} \delta(d) \alpha(a) \payd(d, a)
$
and 
$
\Paya(\delta,\alpha) 
\eqdef {\expectDouble{d\leftarrow\delta}{a\leftarrow\alpha}  \paya(d, a)}
= \sum_{\substack{d\in\cald\\ a\in\cala}} \delta(d) \alpha(a) \paya(d, a)
$.

A defender's mixed strategy $\delta\in\distr\cald$ is a \emph{best response} to 
an attacker's mixed strategy $\alpha\in\distr\cala$ if 
$\Payd(\delta, \alpha) = \max_{\delta'\in\distr\cald}\Payd(\delta', \alpha)$.
Symmetrically, $\alpha\in\distr\cala$ is a \emph{best response} to 
$\delta\in\distr\cald$ if 
$\Paya(\delta, \alpha) = \max_{\alpha'\in\distr\cala}\Payd(\delta, \alpha')$.
A \emph{mixed-strategy Nash equilibrium} is a  profile $(\delta^*, \alpha^*)$ such 
that $\delta^*$ is the best response to $\alpha^*$ and vice versa. 
This means that in a Nash equilibrium, no unilateral deviation by any single 
player provides better payoff to that player.
If $\delta^*$ and $\alpha^*$ are point distributions 
concentrated on some $d^*\in\cald$ and $a^*\in\cala$ respectively, 
then $(\delta^*, \alpha^*)$ is a \emph{pure-strategy Nash equilibrium}, 
and will be denoted by $(d^*, a^*)$. 
While not all games have a pure strategy Nash equilibrium, every finite game 
has a mixed strategy Nash equilibrium.

\subsubsection{Sequential games}

In a sequential game players may take turns in choosing their
actions.
In this paper, we only consider the case in which each player 
moves only once, in such a way that one of the players (\emph{the leader}) 
chooses his action first, and commits to it, before the other player 
(\emph{the follower}) makes his choice. 
The follower may have total knowledge of the choice made by the leader, or only partial. 
We refer to the two scenarios by the terms \emph{perfect} and \emph{imperfect information}, 
respectively.

We now give the precise definitions assuming that the leader is the defender.  
The case in which the leader is the attacker is similar. 

A \emph{defender-first sequential game with perfect information} is a 
tuple $(\cald,\, \allowbreak \calda,\, \allowbreak \payd, \allowbreak \paya)$ where $\cald$, $\cala$, $\payd$ and $\paya$ are
defined as in simultaneous games. 
Also the strategies of the defender (the leader) are defined as in simultaneous games:  
an  action $d\in \cald$ for the pure case, and a  distribution $\delta\in\distr\cald$  for the mixed one. 
On the other hand, a  pure strategy for the attacker is a function $\psa:\calda$, which represents the fact that 
his choice of an action $\psa$ in $\cala$ depends on the 
defender's choice $d$.
An attacker's mixed strategy is a probability distribution  $\msa\in\distr(\calda)$ over his pure strategies.\footnote{The definition of the mixed strategies as $\distr(\cald \rightarrow \cala)$ means that the attacker 
draws a function $\psa:\calda$ \emph{before} he knows the choice of the defender. In contrast, the so-called  \emph{behavioral strategies}   are defined as functions $\cald \rightarrow \distr\cala$, and formalize the idea that the draw is made  after the attacker knows such choice. In our setting,  these two definitions are equivalent, in the sense that they yield the same payoff.
}
The defender's and the attacker's \emph{expected payoff functions} for mixed strategies are defined, respectively, as
$
\Payd(\delta,\msa) 
\eqdef{\expectDouble{d\leftarrow\delta}{\psa\leftarrow\msa}
  \payd(d, \psa(d))}
= \sum_{\substack{d\in\cald\\ \psa:\calda}} \delta(d) \msa(\psa) \payd(d, \psa(d))
$
and
$
\Paya(\delta,\msa)
\eqdef{\expectDouble{d\leftarrow\delta}{\psa\leftarrow\msa} \paya(d, \psa(d))}
= \sum_{\substack{d\in\cald\\ \psa:\calda}} \delta(d) \msa(\psa) \paya(d, \psa(d))
$.

The case of imperfect information 
is typically formalized by assuming an \emph{indistinguishability (equivalence) relation} over the actions chosen 
by the leader, representing a scenario in which the follower cannot distinguish between the actions belonging to the same equivalence class. 
The pure strategies of the followers, therefore,  are functions from the set of the equivalence classes on the actions of the leader to his own actions. Formally,  a \emph{defender-first  sequential game with imperfect information } is a
tuple $(\cald,\, \infoseta\rightarrow\cala,\, \payd, \paya)$ where $\cald$, $\cala$, $\payd$ and $\paya$ are
defined as in simultaneous games, and $\infoseta$ is a partition of $\cald$.  
The \emph{expected payoff functions} are defined as before, except that now the argument of 
$\psa$ is the equivalence class of $d$. 
Note that in the case in which all defender's actions are indistinguishable from each other at the eyes of the attacker  (\emph{totally imperfect information}), we have  $\infoseta = \{\cald\}$ and the expected payoff
functions coincide with those of the simultaneous games.

\subsubsection{Zero-sum games and Minimax Theorem}

A game $(\cald,\, \cala,\, \payd, \paya)$ is \emph{zero-sum} if for any $d\in\cald$ and any $a\in\cala$, the defender's loss is equivalent to the attacker's gain, i.e., $\payd(d, a) = -\paya(d, a)$.
For brevity, in zero-sum games we denote by $u$ the attacker's payoff function $\paya$, and by $U$ the attacker's expected payoff $\Paya$.\footnote{Conventionally in game theory  the payoff $u$ is  set to 
be that of the first player, but we prefer to look at the payoff from the point of view of the attacker to be in line with the definition of payoff as vulnerability.}
Consequently,  the goal of the defender is to minimize $U$, and the goal of the attacker is to maximize it. 

In simultaneous zero-sum games the Nash equilibrium corresponds to the solution of the \emph{minimax} problem (or equivalently,  the \emph{maximin} problem), namely, the strategy profile $(\delta^*, \alpha^*)$ such that 
$U(\delta^*, \alpha^*)=\min_{\delta} \max_{\alpha} U(\delta, \alpha)$. 
The  von Neumann's minimax theorem, in fact, ensures that such solution (which always exists) is stable.

\begin{restatable}[von Neumann's minimax theorem]{Theorem}{res:vonneumann}
\label{theo:vonneumann}
Let $\calx \subset \reals^m$ and $\caly \subset \reals^n$ be compact convex sets,
and $\Pay: \calx\times\caly\rightarrow\reals$ be a continuous function such that
$\Pay(x, y)$ is a convex function in $x\in\calx$ and a concave function in $y\in\caly$.
Then
$
\min_{x\in\calx} \max_{y\in\caly} \Pay(x, y) = \max_{y\in\caly} \min_{x\in\calx} \Pay(x, y)
$.
\end{restatable}

A related property is that, under the conditions of Theorem~\ref{theo:vonneumann}, there exists a \emph{saddle point} $(x^*, y^*)$ s.t., for all $x{\in}\calx$  and $y{\in}\caly$: 
$
 \Pay(x^{*}, y){\leq}\Pay(x^*, y^*){\leq}\Pay(x, y^*) 
$.

The solution of the minimax problem can be obtained by using convex optimization techniques. 
In case $\Pay(x, y)$ is affine in $x$ and in $y$, we can also use linear optimization. 

In case $\cald$ and $\cala$ contain two elements each, there is a closed form for the solution.  
Let  $\cald =\{d_0,d_1\}$ and $\cala=\{a_0,a_1\}$ respectively. 
Let $u_{ij}$ be the utility of the defender on $d_i, a_j$.           
Then the Nash equilibrium $(\delta^*,\alpha^*)$ is given by:
$
\delta^*(d_0)=\nicefrac{(u_{11} - u_{10})}{(u_{00} - u_{01} - u_{10} +u_{11})}
$
and
$
\alpha^*(a_0)=\nicefrac{(u_{11} - u_{01})}{(u_{00} - u_{01} - u_{10} +u_{11})}
$
if these values are in $[0,1]$.  
Note that, since there are only two elements, the strategy $\delta^*$ is completely specified by its value in $d_0$, and analogously for $\alpha^*$.
\subsection{Quantitative information flow}
\label{subsec:qif}

Finally, we briefly review the standard framework of
quantitative information flow, which is concerned with
measuring the amount of information leakage in a 
(computational) system.
  
\paragraph{Secrets and vulnerability}
A \emph{secret} is some piece of sensitive information the 
defender wants to protect, such as a user's password, social 
security number, or current location. 
The attacker usually only has some partial knowledge 
about the value of a secret, represented as a probability
distribution on secrets called a \emph{prior}.
We denote by $\calx$ the set of possible secrets,
and we typically use $\pi$ to denote a prior belonging to 
the set $\dist{\calx}$ of probability distributions over 
$\calx$. 

The  \emph{vulnerability} of a secret is a measure of the utility that it represents for the attacker. 
In this paper we consider a very general notion of  vulnerability, following~\cite{Alvim:16:CSF}, and we define 
a vulnerability $\vf$ to be any continuous and convex function of type $\dist{\calx} \rightarrow \reals$.
It has been shown in~\cite{Alvim:16:CSF} 
that these functions coincide
with the set of $g$-vulnerabilities, 
and are, in a precise sense, the most general 
information measures w.r.t. a set of basic 
axioms.~\footnote{More precisely, if posterior vulnerability 
is defined as the expectation of the vulnerability of posterior
distributions, 
the measure respects the data-processing inequality 
and always yields non-negative leakage iff
vulnerability is convex.}

\paragraph{Channels, posterior vulnerability, and leakage}
Computational systems can be modeled as information
theoretic channels.
A 
\emph{channel}
$C : \calx \times \caly \rightarrow \reals$ is a function
in which $\calx$ is a set of \emph{input values}, $\caly$ is a set 
of \emph{output values}, and $C(x,y)$ represents the conditional 
probability of the channel producing output $y \in \caly$ when 
input $x \in \calx$ is provided. 
Every channel $C$ satisfies $0 \leq C(x,y) \leq 1$ for all 
$x\in\calx$ and $y\in\caly$, and $\sum_{y\in\caly} C(x,y) = 1$ for all $x\in\calx$.

A distribution $\pi\in\dist{\calx}$ and a channel $C$ 
with inputs $\calx$ and outputs $\caly$ induce a joint distribution
$p(x,y) = \pi(x)C({x,y})$ on $\calx \times \caly$,
producing joint random variables $X, Y$ with marginal 
probabilities $p(x) = \sum_{y} p(x,y)$ and 
$p(y) = \sum_{x} p(x,y)$, and conditional probabilities 
$p(x{\mid}y) = \nicefrac{p(x,y)}{p(y)}$ if $p(y)\neq 0$. 
For a given $y$ (s.t. $p(y)\neq 0$), the conditional 
probabilities $p(x{\mid}y)$ for each $x \in \calx$ form the 
\emph{posterior distribution $p_{X \mid y}$}.

A channel $C$ in which $\calx$ is a set of secret values 
and $\caly$ is a set of observable values produced
by a system can be used to model computations on secrets.
Assuming the attacker has prior knowledge $\pi$ about
the secret value, knows how a channel $C$ works, and
can observe the channel's outputs, the effect of the channel 
is to update the attacker's knowledge from $\pi$ to a 
collection of posteriors $p_{X \mid y}$, each occurring 
with probability $p(y)$.%

Given a vulnerability $\vf$, a prior $\pi$, and a channel $C$, 
the \emph{posterior vulnerability} $\postvf{\pi}{C}$ is the vulnerability 
of the secret after the attacker has observed the output of the channel $C$.
Formally:
$\postvf{\pi}{C} \eqdef \sum_{y \in \caly} p(y) \priorvf{p_{X \mid y}}$.

It is known from the literature~\cite{Alvim:16:CSF} 
that  the posterior vulnerability is a convex function of $\pi$. 
Namely, for any channel $C$, any family of distributions $\{\pi_i\}$, and any set of convex coefficients $\{c_i\}$, we have: 
$\postvf{\sum_i c_i \pi_i}{C}
\leq \sum_{i} c_i \postvf{\pi_i}{C}$.

The \emph{(information) leakage} of 
channel $C$ under prior $\pi$ is a comparison between 
the vulnerability of the secret before the system
was run---called \emph{prior vulnerability}---and the 
posterior vulnerability of the secret.
Leakage reflects by how much the observation of 
the system's outputs increases the attacker's 
information about the secret. 
It can be defined either 
\emph{additively} ($\postvf{\pi}{C}-\priorvf{\pi}$), or
\emph{multiplicatively} ($\nicefrac{\postvf{\pi}{C}}{\priorvf{\pi}}$).

\section{An illustrative example}
\label{sec:running-example}
\begin{wrapfigure}{r}{0.31\linewidth}
\vspace{-8mm}
\begin{footnotesize}
\noindent 
\texttt{\underline{Program 0}}\\[1mm]
\noindent \texttt{\textbf{High Input:}} $x \in \{0,1\}$\\
\noindent \texttt{\textbf{Low Input:}} $a \in \{0,1\}$\\
\texttt{\textbf{Output:}} $y\in \{0,1\}$\\[0.5mm]
$y= x \cdot a$\\
\textbf{return} $y$\\[2mm]
\noindent 
\texttt{\underline{Program 1}}\\[1mm]
\noindent \texttt{\textbf{High Input:}} $x \in \{0,1\}$\\
\noindent \texttt{\textbf{Low Input:}} $a \in \{0,1\}$\\
\texttt{\textbf{Output:}} $y\in \{0,1\}$\\[0.5mm]
$c \samplefrom{\text{flip coin with bias $\nicefrac{a}{3}$}}$\\
\textbf{if} $c = \mathit{heads}$ \{$y = x$\} \\
\textbf{else} \{$y = \bar{x}$\}\\
\textbf{return} $y$\\
\vspace{-6mm}
\caption{Running example.}
\label{fig:running-exa}
\end{footnotesize}
\vspace{-6mm}
\end{wrapfigure}
We introduce an example which will serve as running example through the paper. 
Although admittedly contrived, this example is simple and yet produces
different leakage measures for all different combinations of visible/invisible choice and simultaneous/sequential games, 
thus providing a way to compare all different scenarios we are interested in. 

Consider that a binary secret must be processed by a program.
As usual, a defender wants to protect the secret value, 
whereas an attacker wants to infer it by observing the 
system's output.
Assume the defender can choose which among two
alternative versions of the program to run.
Both programs take the secret value $x$ as high input,  
and a binary low input $a$ whose value is chosen by the attacker. 
They both return the output in a low variable $y$.~\footnote{We adopt the usual 
convention in QIF of referring to secret variables, inputs and outputs in programs 
as \emph{high}, and to their observable counterparts as \emph{low}.}
\texttt{Program 0} returns the binary product of $x$ and $a$,
whereas \texttt{Program 1} flips a coin with bias $\nicefrac{a}{3}$
(i.e., a coin which returns heads with probability $\nicefrac{a}{3}$)
and returns $x$ if the  result is
heads, and the complement $\bar{x}$ of $x$ otherwise.
The two programs are represented in Figure~\ref{fig:running-exa}.

The combined choices of the defender's and of the
attacker's determine how the system behaves.
Let $\cald{=} \{0,1\}$ represent the set of the defender's
choices---i.e., the index of the program to use---, and
$\cala = \{0,1\}$ represent the set of the attacker's
choices---i.e., the value of the low input $a$. 
We shall refer to the elements of $\cald$ and $\cala$ as \emph{actions}.
For each possible combination of actions 
$d \in \cald$ and $a \in \cala$, we can construct a channel 
$C_{da}$ modeling how the resulting system behaves.
Each channel $C_{da}$ is a function of type 
$\calx \times \caly \rightarrow \reals$, where 
$\calx = \{0,1\}$ is the set of possible high input values 
for the system, and $\caly = \{0,1\}$ is the set of possible 
output values from the system.
Intuitively, each channel provides the probability that the
system (which was fixed by the defender) produces output 
$y\in\caly$ given that the high input is $x \in \calx$
(and that the low input was fixed by the attacker). 
The four possible channels are depicted as matrices below.
\begin{small}
\begin{align*}
\begin{array}{|c|c|c|}
\hline
C_{00} & y=0 & y=1 \\ \hline
x=0    & 1 & 0 \\
x=1    & 1 & 0 \\ \hline
\end{array}
\quad
\begin{array}{|c|c|c|}
\hline
C_{01} & y=0 & y=1 \\ \hline
x=0    & 1  & 0  \\
x=1    & 0  & 1  \\ \hline
\end{array}
\quad
\begin{array}{|c|c|c|}
\hline
C_{10} & y=0 & y=1 \\ \hline
x=0    & 0 & 1 \\
x=1    & 1 & 0 \\ \hline
\end{array}
\quad
\begin{array}{|c|c|c|}
\hline
C_{11} & y=0 & y=1 \\ \hline
x=0    & \nicefrac{1}{3} & \nicefrac{2}{3} \\
x=1    & \nicefrac{2}{3} & \nicefrac{1}{3} \\ \hline
\end{array}
\end{align*}
\end{small}

Note that channel $C_{00}$ does not leak any 
information about the input $x$ (i.e., it is 
\emph{non-interferent}), whereas channels $C_{01}$ and 
$C_{10}$ completely reveal $x$.
Channel $C_{11}$ is an intermediate case: it leaks some 
information about $x$, but not all.

We want to investigate how the defender's and the attacker's choices
influence the leakage of the system. 
For that we can just consider the (simpler) notion of posterior vulnerability, 
since in order to make the comparison fair we need to assume that the prior is 
always the same in the various scenarios, and this implies that the 
leakage is in a one-to-one correspondence with the posterior vulnerability 
(this happens for both additive and multiplicative leakage). 

\begin{wraptable}{r}{0.42\linewidth}
\centering
\begin{small}
\renewcommand{\arraystretch}{1.15}
\vspace{-9mm}
\[
\begin{array}{|c|c|c|}
\hline
\vf & a=0 & a=1 \\ \hline
d=0    & \nicefrac{1}{2} & 1 \\ 
\hline
d=1    & 1 &  \nicefrac{2}{3}\\ \hline
\end{array}
\]
\renewcommand{\arraystretch}{1}
\vspace{-6mm}
\end{small}
\caption{{\rm Vulnerability of each channel $C_{da}$ in the running example.} }
\label{table:vulnerabilitygame}
\vspace{-6mm}
\end{wraptable}
For this example, assume we are interested in 
Bayes vulnerability~\cite{Chatzikokolakis:08:JCS,Smith:09:FOSSACS}, defined as
$\vf(\pi) = \max_{x} \pi(x)$ for every $\pi \in \dist{\calx}$.
Assume for simplicity that 
the prior is the uniform prior $\pi_u$.
In this case we know from \cite{Braun:09:MFPS} that the posterior 
Bayes vulnerability of a channel is the sum of the greatest elements 
of each column, divided by the total number of inputs. 
Table~\ref{table:vulnerabilitygame} provides the
Bayes vulnerability $\vf_{da} \eqdef \postvf{\pi_{u}}{C_{da}}$
of each channel considered above.

Naturally, the attacker aims at maximizing the vulnerability of the system, 
while the defender tries to minimize it. 
The resulting vulnerability will depend on various factors, in particular on
whether the two players make their choice \emph{simultaneously} 
(i.e. without knowing the choice of the opponent) or \emph{sequentially}. 
Clearly, if the choice of a player who moves first is known by an opponent who moves second, 
the opponent will be in advantage. 
In the above example, for instance, if the defender knows the choice $a$ of the attacker, the most convenient choice for him is to set  $d=a$, and 
the vulnerability will be at most $\nicefrac{2}{3}$ . 
Vice versa, if the attacker knows the choice $d$  of the defender, the most convenient choice for him is to set  $a\neq d$. The vulnerability in this case will be $1$.

Things become more complicated when players make choices simultaneously. 
None of the pure choices of $d$ and $a$ are the best for the corresponding player, 
because the vulnerability of the system depends also on the (unknown) choice of the other player. 
Yet there is a strategy leading to the best possible situation for both players (the \emph{Nash equilibrium}), but it is mixed (i.e., probabilistic), in that the players randomize their choices according to some precise distribution. 

Another factor that affects vulnerability is whether or not the defender's choice is known to the attacker at the moment in which he observes the output of the channel. 
Obviously, this corresponds to whether or not the attacker knows what channel he is observing. 
Both cases are plausible: naturally the defender has all the interest in keeping his choice 
(and, hence, the channel used) secret, since then the attack will be less effective 
(i.e., leakage will be smaller). 
On the other hand, the attacker may be able to identify the channel used anyway, for instance 
because the two programs have different running times. 
We will call these two cases \emph{hidden} and \emph{visible} choice, respectively. 

It is possible to model players' strategies, as well as hidden and visible choices, 
as operations on channels. 
This means that we can look at the whole system as if it were a single channel, 
which will turn out to be useful for some proofs of our technical results.
Next section is dedicated to the definition of these operators. 
We will calculate the exact values for our example in Section~\ref{sec:games-setup}.

\section{Visible and hidden choice operators on channels}
\label{sec:operators}
In this section we define matrices and 
some basic operations on them.
Since channels are a particular kind of 
matrix, we use these matrix operations to define
the operations of visible and hidden choice among
channels, and to prove important properties of 
these channel operations.

\subsection{Matrices, and their basic operators}
\label{sec:matrices-basics}

Given two  sets $\calx$ and $\caly$, a \emph{matrix} 
is a total function of type $\calx \times \caly \rightarrow \reals$.
Two matrices 
$M_{1}: \calx_{1} \times \caly_{1} \rightarrow \reals$ and 
$M_{2}: \calx_{2} \times \caly_{2} \rightarrow \reals$ 
are said to be \emph{compatible} if $\calx_{1} = \calx_{2}$.
If it is also the case that $\caly_{1} = \caly_{2}$, we say 
that the matrices \emph{have the same type}.
The \emph{scalar multiplication} $r{\cdot}M$ between a scalar $r$ and
a matrix $M$ is defined as usual, and so is the \emph{summation}
$\left(\sum_{i \in \cali} M_{i}\right)(x,y) = M_{i_{1}}(x,y) \add \ldots \add M_{i_{n}}(x,y)$
of a family $\{M_{i}\}_{i \in \cali}$ of matrices all of a same type.

Given a family $\{M_{i}\}_{i \in \cali}$ of
compatible matrices s.t. each $M_{i}$ has type 
$\calx \times \caly_{i} \rightarrow \reals$, their \emph{concatenation} $\bigconc_{i \in \cali}$ 
is the matrix having all columns of every matrix in the family,
in such a way that every column is tagged with the matrix it
came from.
Formally, 
$\left( \bigconc_{i \in \cali} M_{i} \right)(x,(y,j)) = \,M_{j}(x,y)$, 
if $y \in \caly_{j}$,
and the resulting matrix has type
$\calx \times \left( \bigsqcup_{i \in \cali} \caly_{i} \right) \rightarrow \reals$.~\footnote{%
$\bigsqcup_{i \in \cali} \caly_{i} = \caly_{i_{1}} \sqcup \caly_{i_{2}} \sqcup \ldots \sqcup \caly_{i_{n}}$ denotes the \emph{disjoint union} 
$\{ (y,i) \mid y \in \caly_{i}, i \in \cali \}$
of the sets $\caly_{i_{1}}$, $\caly_{i_{2}}$, $\ldots$, $\caly_{i_{n}}$.
}
When the family $\{M_{i}\}$ has only two elements
we may use the \emph{binary} version $\conc$ of the
concatenation operator.
The following 
depicts the concatenation of two matrices $M_{1}$ and $M_{2}$ in tabular form.
$$
\begin{small}
\begin{array}{|c|cc|}
\hline
M_{1} & y_{1} & y_{2} \\ \hline
x_{1} & 1 & 2 \\
x_{2} & 3 & 4 \\ \hline
\end{array}
\conc
\begin{array}{|c|ccc|}
\hline
M_{2} & y_{1} & y_{2} & y_{3} \\ \hline
x_{1} & 5 & 6 & 7 \\ 
x_{2} & 8 & 9 & 10 \\ \hline
\end{array} 
\;=\;
\begin{array}{|c|ccccc|}
\hline
M_{1} \conc M_{2} & (y_{1},1) & (y_{2},1) & (y_{1},2) & (y_{2},2) & (y_{3},2) \\ \hline
x_{1} & 1 & 2 & 5 & 6 & 7 \\ 
x_{2} & 3 & 4 & 8 & 9 & 10 \\ \hline
\end{array}
\end{small}
$$

\subsection{Channels, and their hidden and visible choice operators}
\label{sec:operators-definition}

A channel 
is a \emph{stochastic} matrix, i.e., 
all elements are non-negative, and all rows sum up to 1.
Here we will define two operators specific for channels.
In the following, for any real value $0 \leq p \leq 1$, 
we denote by $\bar{p}$ the value $1 - p$.

\subsubsection{Hidden choice}
The first operator models a hidden probabilistic 
choice among channels.
Consider a family $\left\{ C_{i} \right\}_{i \in \cali}$ 
of channels of a same type.
Let $\mu \in \dist{\cali}$ be a probability
distribution on the elements of the index set $\cali$.
Consider an input $x$ is fed to one of
the channels in $\left\{ C_{i} \right\}_{i \in \cali}$,
where the channel is randomly picked according to $\mu$.
More precisely, an index $i \in \cali$ is sampled with 
probability $\mu(i)$,  then the input
$x$ is fed to channel $C_{i}$,
and the output $y$ produced by the channel is then made visible, 
but not the index $i$ of the channel that was used.
Note that we consider hidden choice only among
channels of a same type:  if the sets of outputs
were not identical, the produced output 
might implicitly reveal which channel was used.

Formally, given a family $\{C_{i}\}_{i \in \cali}$ of 
channels s.t. each $C_{i}$ has same type 
$\calx \times \caly \rightarrow \reals$, the 
\emph{hidden choice operator} $\HChoice{i}{\mu}$ 
is defined as 
$
\HChoice{i}{\mu} C_{i} = \sum_{i \in \cali} \mu(i) \,C_{i}
$.

\begin{restatable}{Proposition}{restypehiddenchoice}
\label{prop:type-hidden-choice}
Given a family $\{C_{i}\}_{i \in \cali}$ of 
channels of  type $\calx \times \caly \rightarrow \reals$, 
and a distribution $\mu$ on $\cali$, 
the hidden choice 
$\HChoice{i}{\mu} C_{i}$ is 
a channel of type
$\calx \times \caly \rightarrow \reals$.
\end{restatable}

In the particular case in which the family $\{C_{i}\}$ 
has only two elements $C_{i_{1}}$ and $C_{i_{2}}$, the
distribution $\mu$ on indexes is completely determined 
by a real value $0 \leq p \leq 1$ s.t. 
$\mu(i_{1}) = p$ and $\mu(i_{2}) = \bar{p}$.
In this case we may use the \emph{binary} version $\hchoice{p}$
of the hidden choice operator:
$
C_{i_{1}}{\hchoice{p}}C_{i_{2}} = p\,C_{i_{1}}{+}\bar{p}\,C_{i_{2}}
$.
%
The example 
below depicts the hidden choice
between channels $C_{1}$ and $C_{2}$, with probability
$p{=}\nicefrac{1}{3}$.
$$
\begin{small}
\begin{array}{|c|cc|}
\hline
C_{1} & y_{1} & y_{2} \\ \hline
x_{1} & \nicefrac{1}{2} & \nicefrac{1}{2} \\
x_{2} & \nicefrac{1}{3} & \nicefrac{2}{3} \\ \hline
\end{array}
\hchoice{\nicefrac{1}{3}}
\begin{array}{|c|cc|}
\hline
C_{2} & y_{1} & y_{2} \\ \hline
x_{1} & \nicefrac{1}{3} & \nicefrac{2}{3} \\
x_{2} & \nicefrac{1}{2} & \nicefrac{1}{2} \\ \hline
\end{array}
\; = \;
\begin{array}{|c|cc|}
\hline
C_{1} \hchoice{\nicefrac{1}{3}} C_{2} & y_{1} & y_{2} \\ \hline
x_{1} & \nicefrac{7}{18} & \nicefrac{11}{18} \\ 
x_{2} & \nicefrac{4}{9} & \nicefrac{5}{9} \\ \hline
\end{array}
\end{small}
$$

\subsubsection{Visible choice}
The second operator models a visible probabilistic 
choice among channels.
Consider a family $\left\{ C_{i} \right\}_{i \in \cali}$ 
of compatible channels.
Let $\mu \in \dist{\cali}$ be a probability
distribution on the elements of the index set $\cali$.
Consider an input $x$ is fed to one of
the channels in $\left\{ C_{i} \right\}_{i \in \cali}$,
where the channel is randomly picked according to $\mu$.
More precisely, an index $i \in \cali$ is sampled with 
probability $\mu(i)$,  then the input
$x$ is fed to channel $C_{i}$, and the output $y$ produced by the channel is then made visible, 
along with the index $i$ of the channel that was used.
Note that visible choice makes sense only 
between compatible channels,
but it is not required that the output set of each channel 
be the same.

Formally, given  $\{C_{i}\}_{i \in \cali}$ of
compatible channels s.t. each $C_{i}$ has type 
$\calx \times \caly_{i} \rightarrow \reals$, 
and a distribution $\mu$ on $\cali$, the 
\emph{visible choice operator} $\VChoice{i}{\mu}$ 
is defined as 
$
\VChoice{i}{\mu} C_{i} = \bigconc_{i \in \cali} \;\mu(i)\, C_{i}
$.

\begin{restatable}{Proposition}{restypevisiblechoice}
\label{prop:type-visible-choice}
Given a family $\{C_{i}\}_{i \in \cali}$ of 
compatible channels s.t. each $C_{i}$ has type 
$\calx \times \caly_{i} \rightarrow \reals$,
and a distribution $\mu$ on $\cali$, 
the result of the visible choice 
$\VChoice{i}{\mu} C_{i}$ is a channel
of type 
$\calx \times \left( \bigsqcup_{i \in \cali} \caly_{i} \right) \rightarrow \reals$.
\end{restatable}

In the particular case the family $\{C_{i}\}$ 
has only two elements $C_{i_{1}}$ and $C_{i_{2}}$, the
distribution $\mu$ on indexes is completely determined 
by a real value $0 \leq p \leq 1$ s.t. 
$\mu(i_{1}) = p$ and $\mu(i_{2}) = \bar{p}$.
In this case we may use the \emph{binary} version $\vchoice{p}$
of the visible choice operator:
$
C_{i_{1}} \vchoice{p} C_{i_{2}} = p\, C_{i_{1}} \conc \bar{p}\, C_{i_{2}}
$.
%
The following 
depicts the visible choice 
betwee channels $C_{1}$ and $C_{3}$, with probability
$p{=}\nicefrac{1}{3}$.
$$
\begin{small}
\begin{array}{|c|cc|}
\hline
C_{1} & y_{1} & y_{2} \\ \hline
x_{1} & \nicefrac{1}{2} & \nicefrac{1}{2} \\
x_{2} & \nicefrac{1}{3} & \nicefrac{2}{3} \\ \hline
\end{array}
\vchoice{\nicefrac{1}{3}}
\begin{array}{|c|cc|}
\hline
C_{3} & y_{1} & y_{3} \\ \hline
x_{1} & \nicefrac{1}{3} & \nicefrac{2}{3} \\
x_{2} & \nicefrac{1}{2} & \nicefrac{1}{2} \\ \hline
\end{array}
\; = \;
\begin{array}{|c|cccc|}
\hline
C_{1} \vchoice{\nicefrac{1}{3}} C_{3} & (y_{1},1) & (y_{2},1) & (y_{1},3) & (y_{3},3) \\ \hline
x_{1} & \nicefrac{1}{6} & \nicefrac{1}{6} & \nicefrac{2}{9} & \nicefrac{4}{9} \\ 
x_{2} & \nicefrac{1}{9} & \nicefrac{2}{9} & \nicefrac{1}{3} & \nicefrac{1}{3} \\ \hline
\end{array}
\end{small}
$$

\subsection{Properties of hidden and visible choice operators}
\label{sec:operators-properties}

We now prove algebraic properties of channel operators.
These properties will be useful when we model a (more complex) 
protocol as the composition of smaller channels via hidden or visible 
choice.

Whereas the properties of hidden choice 
hold generally with equality, those of 
visible choice are subtler.
For instance, visible choice is not idempotent,
since in general $C{\vchoice{p}}C \neq C$.
(In fact if $C$ has type $\calx \times \caly \rightarrow \reals$,
$C{\vchoice{p}}C$ has type $\calx \times (\caly \sqcup \caly) \rightarrow \reals$.)
However, idempotency and other properties involving visible
choice hold if we replace the notion of 
equality with the more relaxed notion of \qm{equivalence} 
between channels.
Intuitively, two channels are equivalent if they
have the same input space and yield 
the same value of vulnerability for every 
prior and every vulnerability function.

\begin{Definition}[Equivalence of channels]
\label{def:equivalence-channels}
Two compatible channels $C_{1}$ and $C_{2}$
with domain $\calx$
are \emph{equivalent}, denoted by $C_{1} \equiv C_{2}$,
if for every prior $\pi \in \dist{\calx}$ and every 
posterior vulnerability $\vf$ we have
$
\postvf{\pi}{C_{1}} = \postvf{\pi}{C_{2}}
$.
\end{Definition}

Two equivalent channels are indistinguishable from the point of view of information leakage, and in most cases we can just identify them. 
Indeed, nowadays there is a tendency to use \emph{abstract channels}~\cite{McIver:14:POST,Alvim:16:CSF}, which capture exactly 
the important behavior with respect to any  form of leakage. 
In this paper, however, we cannot use abstract channels because the hidden choice operator needs a concrete representation in order to be defined
unambiguously. 

The first properties we prove regard idempotency of operators, which
can be used do simplify the representation of some protocols.

\begin{restatable}[Idempotency]{Proposition}{residempotency}
\label{prop:idempotency}
Given a family $\{C_{i}\}_{i \in \cali}$ of 
channels s.t. $C_{i} = C$ for all $i \in \cali$,
and a distribution $\mu$ on $\cali$, then:
(a) ${\HChoice{i}{\mu}C_{i} = C}$; and
(b) ${\VChoice{i}{\mu}C_{i} \equiv C}$.
\end{restatable}

The following properties regard the reorganization of operators, 
and they will be essential in some technical results in which we invert 
the order in which hidden and visible choice are applied in a protocol.

\begin{restatable}[\qm{Reorganization of operators}]{Proposition}{resreorganizationoperators}
\label{prop:reorganization-operators}
Given a family $\{C_{i j}\}_{i \in \cali, j \in \calj}$ of 
channels indexed by sets $\cali$ and $\calj$,
a distribution $\mu$ on $\cali$, and
a distribution $\eta$ on $\calj$:

\begin{enumerate}[(a)]
\item 
${\HChoice{i}{\mu} \;\HChoice{j}{\eta} C_{i j} = \HChoiceDouble{i}{\mu}{j}{\eta} C_{i j}}$, if all $C_{i}$'s have the same type;\\

\item ${\VChoice{i}{\mu} \; \VChoice{j}{\eta} C_{i j} \equiv \VChoiceDouble{i}{\mu}{j}{\eta} C_{i j}}$, if all $C_{i}$'s are compatible; and\\

\item \label{item:c}${\HChoice{i}{\mu} \; \VChoice{j}{\eta} C_{i j} \equiv \VChoice{j}{\eta}\; \HChoice{i}{\mu} C_{i j}}$, 
if, for each $i$, all $C_{i j}$'s have same type $\calx{\times}\caly_{j}{\rightarrow}\reals$.
\end{enumerate}

\end{restatable}

\subsection{Properties of vulnerability w.r.t. channel operators}
\label{sec:convexity-vulnerability}

We now derive some relevant properties of vulnerability w.r.t. our 
channel operators, which will be later used to obtain the Nash equilibria 
in information leakage games with different choice operations.

The first result states that posterior vulnerability is 
convex w.r.t. hidden choice (this result was already presented in \cite{Alvim:17:GameSec}), and 
linear w.r.t. to visible choice.

\begin{restatable}{Theorem}{resconvexVq}
\label{theo:convex-V-q}
Let $\{C_{i}\}_{i \in \cali}$ be a family of channels,
and $\mu$ be a distribution on $\cali$.
Then, for every distribution $\pi$ on $\calx$, and every
vulnerability $\vf$:
\begin{enumerate}[(a)]

\item \label{enumerate:vg:h:convex}
posterior vulnerability is convex w.r.t. to hidden choice:
$
\postvf{\pi}{\HChoice{i}{\mu} C_{i}} \leq \sum_{i \in \cali} \mu(i) \,\postvf{\pi}{C_{i}}
$
if all $C_{i}$'s have the same type.

\item \label{enumerate:vg:v:linear}
posterior vulnerability is linear w.r.t. to visible choice:
$
\postvf{\pi}{\VChoice{i}{\mu} C_{i}} = \sum_{i \in \cali} \mu(i) \,\postvf{\pi}{C_{i}}
$
if all $C_{i}$'s are compatible.
\end{enumerate}
\end{restatable}

The next result is concerned with posterior vulnerability 
under the composition of channels using both operators.

\begin{restatable}{Corollary}{resconcaveconvexV}
\label{cor:concave-convex-V}
Let $\{C_{i j}\}_{i \in \cali, j \in \calj}$ be a family of channels,
all with domain $\calx$ and with the same type, and let $\pi\in \distr\calx$, and $\vf$ be any vulnerability. 
Define
$\Pay: \distr\cali\times\distr\calj\rightarrow \reals$ as follows:
$
\Pay(\mu, \eta) \eqdef \postvf{\pi}{\HChoice{i}{\mu} \; \VChoice{j}{\eta} \; C_{i j}}
$.
Then $\Pay$ is convex on $\mu$ and linear on $\eta$.
\end{restatable}

\section{Information leakage games}
\label{sec:games-setup}
In this section we present our framework for  reasoning about information leakage, 
extending the notion of \emph{information leakage games} proposed in \cite{Alvim:17:GameSec} from only simultaneous games with hidden choice to 
both simultaneous and sequential games, with either hidden or visible choice.

In an information leakage game the defender tries to minimize the leakage of information 
from the system, while the attacker tries to maximize it.
In  this basic scenario, their goals are just opposite (zero-sum). 
Both of them can influence the execution and the observable behavior of the system via a 
specific set of actions.
We assume players to be rational (i.e., they are able to figure out what is 
the best strategy to 
maximize their expected payoff), and that the set of actions and the
payoff function are common knowledge. 

Players choose their own strategy, which in general may be mixed 
(i.e. probabilistic), and choose their action by a random draw 
according to that strategy. 
After both players have performed their actions, the system runs 
and produces some output value which is visible to the attacker and
may leak some information about the secret.
The amount of leakage constitutes the attacker's gain, and the defender's loss. 

To quantify the  leakage we model the system as an information-theoretic channel  
(cf. Section~\ref{subsec:qif}). We recall that leakage is defined as the difference (additive leakage) 
or the ratio (multiplicative leakage) between posterior and prior vulnerability. 
Since we are only interested in comparing the leakage of different channels
for a given prior, 
\emph{we will define the payoff just as the posterior vulnerability},
as the value of prior vulnerability will be the same for every channel.

\subsection{Defining information leakage games}
\label{subsec:def:leakage-game}

An (\emph{information}) \emph{leakage game}
consists of:
(1) two nonempty sets $\cald$, $\cala$ of defender's and attacker's actions  
respectively,
(2) a function $C: \cald{\times}\cala \rightarrow (\calx{\times}\caly \rightarrow \reals)$ 
that associates to each pair of actions  $(d,a)\in\cald{\times}\cala$ a channel 
$C_{da}: \calx{\times}\caly \rightarrow \reals$, 
(3) a prior $\pi \in \dist\calx$ on secrets, and 
(4) a vulnerability measure $\vf$.  
%
The payoff function $u: \cald{\times}\cala \rightarrow\reals$ for pure strategies 
is defined as
$
u(d,a) \eqdef \postvf{\pi}{C_{da}}
$.
We have only one payoff function because the game is zero-sum. 

Like in traditional game theory, the order of actions and the extent by which a player knows the move performed by the opponent play a critical role in deciding strategies and determining the payoff. 
In security, however, knowledge of the opponent's move affects the game in yet another way:
the effectiveness of the attack, i.e., the amount of leakage, depends crucially on whether or not
the attacker knows what channel is being used. 
It is therefore convenient to distinguish two phases in the leakage game: 


\begin{description}
\item[\textbf{Phase 1:}] Each player determines the most convenient strategy (which in general is mixed) for himself, and draws his action accordingly. One of the players may commit first to his action, and his choice may or may not be revealed to the follower. In general, knowledge of the leader's action may help the follower choose a more advantageous strategy. 

\item[\textbf{Phase 2:}]  The attacker observes the output of the selected channel $C_{da}$ and performs his attack on the secret. In case he knows the defender's action, he is able to determine the exact channel $C_{da}$ being used (since, of course, the attacker knows his own action), 
and his payoff will be the posterior vulnerability $\postvf{\pi}{C_{da}}$.
However, if the attacker does not know exactly which channel has been used, then his payoff will be smaller.  
\end{description}

Note that the issues raised in Phase 2 are typical of leakage games; 
they do not have a correspondence (to the best of our knowledge) 
in traditional game theory. 
On the other hand, these issues are central to security, as they 
reflect the principle of preventing the attacker from inferring 
the secret by obfuscating the link between secret and observables. 

Following the above discussion, we consider various possible 
scenarios for games, along two lines of classification.
First, there are three possible orders for the two players' actions.
\begin{description}
\item[{\bf Simultaneous:}]
The players choose (draw) their actions in parallel, each without knowing the choice of the other.
\item[{\bf Sequential, defender-first:}]
The defender draws an action, and commits to it,  before the attacker does.
\item[{\bf Sequential, attacker-first:}]
The attacker draws an action, and commits to it,  before the defender does.
\end{description}
Note that these sequential games may present imperfect information
(i.e., the follower may not know the leader's action).

Second, the visibility of the defender's action during the attack may vary:
\begin{description}
\item[{\bf Visible choice:}] 
The attacker knows the defender's action when he observes the output of the channel, and therefore 
he knows which channel is being used.
Visible choice is modeled by the operator $\vchoiceop$.  
\item[{\bf Hidden choice:}]
The attacker does not know the defender's action when he observes the output of the channel, and therefore in general he does not exactly know which channel is used (although in some special cases he may infer it from the output).  
Hidden choice is modeled by the operator $\hchoiceop$. 
\end{description}

Note that the distinction between sequential and simultaneous games
is orthogonal to that between visible and hidden choice.
Sequential and simultaneous games model whether or not, respectively,
the follower's choice can be affected by knowledge of the leader's action.
This dichotomy captures how knowledge about the other player's actions 
can \emph{help a player choose his own action}.
On the other hand, visible and hidden choice capture whether or not, respectively, the 
attacker is able to fully determine the channel representing the system, 
\emph{once defender and attacker's actions have already been fixed}.
This dichotomy reflects the different \emph{amounts of information leaked} by 
the system as viewed by the adversary.
For instance, in a simultaneous game neither player can choose his action based on 
the choice of the other.
However, depending on whether or not the defender's choice is visible, 
the adversary will or will not, respectively, be able to completely recover 
the channel used, which will affect the amount of leakage.

If we consider also the subdivision of sequential games into perfect and imperfect information, 
there are $10$ possible different combinations. 
Some, however, make little sense. 
For instance,  defender-first sequential game with perfect information 
(by the attacker) 
does not combine naturally with hidden choice
$\hchoiceop$, 
since that would mean that the attacker knows the action of the defender 
and choses his strategy accordingly, but forgets it at the moment of the attack. 
(We assume \emph{perfect recall}, i.e., the players never forget what they have learned.)
Yet other combinations are not interesting, such as the attacker-first sequential game with 
(totally) imperfect information (by the defender), since it coincides with the simultaneous-game case. 
Note that attacker and defender are not symmetric with respect to hiding/revealing their actions $a$ and $d$, since the knowledge of $a$ affects the game only in the usual sense of game theory, while the knowledge of $d$ also affects the computation of the payoff (cf.  \qm{Phase 2} above). 

Table~\ref{table:games} lists the meaningful and interesting combinations. 
In Game V we assume imperfect information: the attacker does not 
know the action chosen by the defender. 
In all the other sequential games we assume that the follower has perfect 
information.
In the remaining of this section, we discuss each game individually, using the example of Section~\ref{sec:running-example} as running example. 

\begin{table}[!tb]
\centering
\begin{small}
\renewcommand{\arraystretch}{1.35}
\begin{tabular}{c|c|c|c|c|}
\multicolumn{2}{}{} & \multicolumn{3}{c}{Order of action} \\ \cline{3-5}
\multicolumn{2}{}{} & \multicolumn{1}{|c|}{\textbf{simultaneous}} & \textbf{defender 1\textsuperscript{st}} & \textbf{attacker 1\textsuperscript{st}} \\ \cline{2-5}
\multirow{2}{*}{\stackanchor{Defender's}{choice}} & \textbf{visible $\vchoiceop$} & Game I & Game II & Game III \\ \cline{2-5}
&  {\textbf{hidden $\hchoiceop$}} & Game IV & Game V & Game VI \\ \cline{2-5}
\end{tabular}
\renewcommand{\arraystretch}{1}
\vspace{3mm}
\caption{{\rm Kinds of games we consider.}
{\rm All sequential games have perfect information, except for game V.}}
\vspace{-8mm}
\label{table:games}
\end{small}
\end{table}

\subsubsection{Game I (simultaneous with visible choice)}

This simultaneous game can be represented by a tuple $(\cald,\, \cala,\, \pay)$. 
As in all games with visible choice $\vchoiceop$, 
the expected payoff $\Pay$ of a mixed strategy profile $(\delta,\alpha)$ 
is defined to be the expected value of $u$, as in traditional game theory: 
$
\Pay(\delta,\alpha) 
\eqdef {\expectDouble{d\leftarrow\delta}{a\leftarrow\alpha}\hspace{-0.5ex}  \pay(d, a)}
=\hspace{-0.5ex} \sum_{\substack{d\in\cald\\ a\in\cala}} \delta(d) \,\alpha(a)\, \pay(d, a)
$,
where we recall that $\pay(d,a) = \postvf{\pi}{C_{da}}$. 

From Theorem~\ref{theo:convex-V-q}(\ref{enumerate:vg:v:linear}) we derive:
$\Pay(\delta,\alpha) 
= \postvf{\pi}{\VChoiceDouble{d}{\delta}{a}{\alpha} C_{da}} 
$.
hence the whole system can be equivalently regarded as 
the channel $\VChoiceDouble{d}{\delta}{a}{\alpha} C_{da}$. 
Still from Theorem~\ref{theo:convex-V-q}(\ref{enumerate:vg:v:linear}) we 
can derive that $\Pay(\delta,\alpha)$ is linear in $\delta$ and $\alpha$. 
Therefore the Nash equilibrium can be computed using the minimax method 
(cf. Section~\ref{subsec:game-theory}). 
\begin{Example}
Consider the example of Section~\ref{sec:running-example} in the setting of Game I.
The Nash equilibrium $(\delta^*,\alpha^*)$ can be obtained using the closed formula 
from Section \ref{subsec:game-theory},
and it is given by
$
\delta^*(0) =\alpha^*(0) =\nicefrac{(\nicefrac{2}{3}-1)}{(\nicefrac{1}{2}-1-1+\nicefrac{2}{3})}=\nicefrac{2}{5}.
$
The corresponding payoff is
$
\Pay(\delta^*,\alpha^*)= \nicefrac{2}{5}\,\nicefrac{2}{5}\,\nicefrac{1}{2}+\nicefrac{2}{5}\,\nicefrac{3}{5} + \nicefrac{3}{5}\,\nicefrac{2}{5} + \nicefrac{3}{5}\,\nicefrac{3}{5}\,\nicefrac{2}{3}= \nicefrac{4}{5}
$.
\end{Example}

\subsubsection{Game II (defender 1\textsuperscript{st} with visible choice)}

This defender-first sequential game can be represented by a tuple $(\cald,\, \calda,\, \pay)$. 
A mixed strategy profile is of the form $(\delta, \msa)$, with $\delta\in\distr\cald$ and $\msa\in\distr(\calda)$, and the corresponding payoff is
$
\Pay(\delta,\msa) 
\eqdef {\expectDouble{d\leftarrow\delta}{\psa\leftarrow\msa} \pay(d, \psa(d))} \allowbreak
= \allowbreak \sum_{\substack{d\in\cald\\ \psa:\calda}} \delta(d)\, \msa(\psa)\, \pay(d, \psa(d))
$,
where  $\pay(d, \psa(d)) = \postvf{\pi}{C_{d\psa(d)}}$. \\[1mm]

Again, from Theorem~\ref{theo:convex-V-q}(\ref{enumerate:vg:v:linear}) we derive:
$
\textstyle
\Pay(\delta,\msa) 
= \postvf{\pi}{\VChoiceDouble{d}{\delta}{\psa}{\msa}  C_{d \psa(d)}} 
$
and hence the system can be expressed as channel 
$\textstyle \VChoiceDouble{d}{\delta}{\psa}{\msa} C_{d \psa(d)}$.
From the same Theorem we also derive that $\Pay(\delta,\msa)$ is 
linear in $\delta$ and $\msa$, so the mutually optimal strategies can 
be obtained again by solving the minimax problem.
In this case, however, the solution is particularly simple, because it 
is  known that there are optimal strategies which are deterministic.
Hence it is sufficient for the defender to find  the action $d$ which 
minimizes $\max_a \pay(d,a)$. 

\begin{Example}
Consider  the example of Section~\ref{sec:running-example} in the setting of Game II.
If the defender chooses $0$ then the attacker chooses $1$. 
If the defender chooses $1$ then the attacker chooses $0$.
In both cases, the payoff is $1$. The game has therefore two solutions, $(0,1)$ and $(1,0)$.
\end{Example}

\subsubsection{Game III (attacker 1\textsuperscript{st} with visible choice)}

This game is also a sequential game, but with the attacker as the leader. 
Therefore it can be represented as tuple of the form $(\calad,\,  \cala,\, \pay)$. 
It is the same as Game II, except that the roles of the attacker and the 
defender are inverted. 
In particular, the payoff of a mixed strategy profile 
$(\msd, \alpha)\in \distr(\calad)\times \distr\cala$ is given by
$
\Pay(\msd,\alpha)  
\eqdef {\expectDouble{\psd\leftarrow\msd}{a\leftarrow\alpha}\hspace{-0.5ex}
  \pay(\psd(a),a)} = {\sum_{\substack{\psd:\calad \\ a\in\cala}}} \msd(\psd) 
  \, \alpha(a)\, \pay(\psd(a), a) \allowbreak 
= \allowbreak {\postvf{\pi}{\VChoiceDouble{\psd}{\msd}{a}{\alpha}  C_{\psd(a) a}} }
$,
%
and  the whole system can be equivalently regarded as channel 
$\VChoiceDouble{\psd}{\msd}{a}{\alpha}  C_{\psd(a) a}$. 
Obviously, also  in this case the minimax problem has  a deterministic solution. 

In summary, in the sequential case, whether the leader is the defender or the attacker
(Games II and III, respectively), the minimax problem has always a deterministic 
solution~\cite{Osborne:94:BOOK}.

\begin{Theorem}
\label{theo:deterministic-strategies}
In a defender-first sequential game with visible choice, there exist 
$d\in\cald$ and $a\in\cala$ such that, 
for every $\delta\in\distr\cald$ and $\msa\in\distr(\calda)$ we have:
$
\Pay(d,\msa)\leq \pay(d,a)\leq \Pay(\delta,a)
$.
Similarly, in an attacker-first sequential game with visible choice, there exist 
$d\in\cald$ and $a\in\cala$ such that, 
for every $\msd\in\distr(\calad)$ and $\alpha\in\distr\cala$ we have:
$
\Pay(d,\alpha)\leq \pay(d,a)\leq \Pay(\msd,a)
$.
\end{Theorem}

\begin{Example}
Consider now the example of Section~\ref{sec:running-example} in the setting of Game III.
If the attacker chooses $0$ then the defender chooses $0$ and the payoff is $\nicefrac{1}{2}$.  
If the attacker chooses $1$ then the defender  chooses $1$ and the payoff is $\nicefrac{2}{3}$.
The latter case is more convenient for the attacker, hence the solution of the game is the strategy profile $(1,1)$. 
\end{Example}

\subsubsection{Game IV (simultaneous with hidden choice)}

This game is a tuple $(\cald,\cala,\pay)$.
However, \emph{it is not an ordinary game} in the sense that 
\emph{the payoff a mixed strategy profile cannot be defined by averaging
the payoff of the corresponding pure strategies}. 
More precisely, the payoff of a mixed profile is defined by
averaging on the strategy of the attacker, but not on that 
of the defender. 
In fact, when hidden choice is used, there is an additional level 
of uncertainty in the relation between the observables and the secret 
from the point of view of the attacker, since he is not sure about which 
channel is producing those observables. 
A mixed strategy $\delta$ for the defender produces a convex combination 
of  channels  (the channels associated to the pure strategies) with the 
same coefficients, and we know from previous sections that  the vulnerability 
is a convex function of the channel, and in general is not linear. 

In order to define the payoff of a mixed strategy  profile  $(\delta,\alpha)$, 
we need therefore to consider the channel that the attacker perceives given his 
limited knowledge. Let us assume that the action that the attacker draws from 
$\alpha$  is $a$. 
He does not know the action of the defender, but we can assume 
that he knows his strategy (each player can derive the optimal strategy of the 
opponent, under the assumption of common knowledge and rational players).

The channel the attacker will see is 
$\HChoice{d}{\delta} C_{d a}$, obtaining a corresponding payoff of 
$\postvf{\pi}{{\HChoice{d}{\delta}} {C_{d a}}}$.
By averaging on the strategy of the attacker we obtain 
$
\Pay(\delta,\alpha) 
\eqdef  
{\expectDouble{a\leftarrow\alpha}{}\hspace{-1ex} \; \postvf{\pi}{{\HChoice{d}{\delta}} {C_{d a}}}}
= \sum_{a\in\cala} \,\alpha(a)\, \postvf{\pi}{{\HChoice{d}{\delta}} {C_{d a}}}
$.
From Theorem~\ref{theo:convex-V-q}(\ref{enumerate:vg:v:linear}) we derive:
$
\Pay(\delta,\alpha) 
= \postvf{\pi}{\VChoice{a}{\alpha}\, {\HChoice{d}{\delta}} {C_{d a}}} 
$
and hence the whole system can be equivalently regarded as channel 
${\VChoice{a}{\alpha}\, {\HChoice{d}{\delta}} {C_{d a}}}$.
Note that, by Proposition~\ref{prop:reorganization-operators}(\ref{item:c}), 
the order of the operators is interchangeable, and the system can be 
equivalently regarded as $ {\HChoice{d}{\delta}}\,{\VChoice{a}{\alpha} {C_{d a}}}$.
This shows the robustness of this model. 

From Corollary~\ref{cor:concave-convex-V} we  derive that $\Pay(\delta,\alpha)$ is convex in $\delta$ and linear in $\eta$, hence we can compute the Nash equilibrium by the minimax method.

\begin{Example}\label{exa:exempio}
Consider now the example of Section~\ref{sec:running-example} in the setting of Game IV.
For $\delta\in\distr\cald$ and $\alpha\in\distr\cala$, let $p=\delta(0)$ and $q=\alpha(0)$. 
The system can be represented by the channel 
$(C_{00}\hchoice{p} C_{10})\vchoice{q}(C_{01}\hchoice{p} C_{11})$
represented below.
\begin{align*}
\begin{small}
\begin{array}{|c|c|c|}
\hline
C_{00}  \hchoice{p} C_{10}& y=0 & y=1 \\ \hline
x=0    & p & \bar{p} \\
x=1    & 1 & 0 \\ \hline
\end{array}
\quad
\vchoice{q}
\quad
\begin{array}{|c|c|c|}
\hline
C_{01}\hchoice{p} C_{11}& y=0 & y=1 \\ \hline
x=0    & \nicefrac{1}{3}+ \nicefrac{2}{3}\;p &   \nicefrac{2}{3}- \nicefrac{2}{3}\;p \\
x=1    & \nicefrac{2}{3}- \nicefrac{2}{3}\;p  & \nicefrac{1}{3}+ \nicefrac{2}{3}\,p \\ \hline
\end{array}
\end{small}
\end{align*}
\noindent For uniform $\pi$, we have
$\postvf{\pi}{ C_{00} \hchoice{p} C_{10}}{=}1-\nicefrac{1}{2}$;
and 
$\postvf{\pi}{ C_{10}  \hchoice{p} C_{11}}$ is equal to
$
\nicefrac{2}{3}- \nicefrac{2}{3}\;p$ if $p\leq \nicefrac{1}{4}$,
and equal to 
$\nicefrac{1}{3}+ \nicefrac{2}{3}\;p$ 
if $p> \nicefrac{1}{4}$.
Hence the payoff, expressed in terms of $p$ and $q$, is
$\Pay(p,q) = q(1-\nicefrac{1}{2}) + \bar{q} (\nicefrac{2}{3} - \nicefrac{2}{3}\;p)$
if $p\leq \nicefrac{1}{4}$, and 
$\Pay(p,q) = q(1-\nicefrac{1}{2}) + \bar{q} (\nicefrac{1}{3} + \nicefrac{2}{3}\;p)$ 
if $p> \nicefrac{1}{4}$.
The Nash equilibrium $(p^*,q^*)$ is given by
$
p^* = \argmin_p\max_q 	\Pay(p,q)
$
and 
$q^* = \argmax_q\min_p 	\Pay(p,q)
$, and by solving the above, we obtain
$p^* = q^*= \nicefrac{4}{7}$.

\end{Example}

\subsubsection{Game V (defender 1\textsuperscript{st} with hidden choice)}

This is a defender-first sequential game with imperfect information, hence it can be represented as a tuple of the form  $(\cald,\, \infoseta\rightarrow\cala,\, \payd, \paya)$, where $\infoseta$ is a partition of $\cald$. 
Since we are assuming perfect recall, and the attacker does not know anything about the 
action chosen by the defender in Phase 2, i.e., at the moment of the attack (except the probability distribution determined by his strategy), we must assume that the attacker does not know anything in Phase 1 either. Hence the indistinguishability relation must be total, i.e.,  $\infoseta=\{\cald\}$. But $\{\cald\}\rightarrow\cala$ is equivalent to $\cala$, hence this kind of game is equivalent to Game IV.

It is also a well known fact in Game theory that when in a sequential game 
the follower does not know the leader's move before making his choice, 
the game is equivalent to a simultaneous game.\footnote{However, one could
argue that, since the defender has already committed, the attacker does not
need to perform the action corresponding to the Nash equilibrium, any 
payoff-maximizing solution would be equally good for him.}

\subsubsection{Game VI (attacker 1\textsuperscript{st} with hidden choice)}

This game is also a sequential game with the attacker as the leader, hence it is a tuple of the form
$(\calad,\,  \cala,\, \pay)$. 
It is  similar to Game III, except that the payoff is convex on the strategy of the defender, instead of linear. 
The payoff of the mixed strategy profile  $(\msd, \alpha)\in \distr(\calad)\times \distr\cala$  is
$
\Pay(\msd,\alpha)  
\eqdef {\expectDouble{a\leftarrow\alpha}{}\hspace{-1ex} \; \postvf{\pi}{{\HChoice{\psd}{\msd}} {C_{\psd(a) a}}}} 
= {\postvf{\pi}{\VChoice{a}{\alpha}\;\HChoice{\psd}{\msd}  C_{\psd(a) a}} }
$,
so the whole system can be equivalently regarded as channel 
$\VChoice{a}{\alpha}\;\HChoice{\psd}{\msd}  C_{\psd(a) a}$.
Also in this case the minimax problem has  a deterministic solution, but only for the attacker.
\begin{Theorem}
In an attacker-first sequential game with hidden choice, 
there exist  $a\in\cala$ and $\delta \in \distr\cald$ such that, for every $\alpha\in\distr\cala$ and $\msd\in\distr(\calad)$ we have that
$
\Pay(\delta,\alpha)\leq \Pay(\delta,a)\leq \Pay(\msd,a)
$.
\end{Theorem}

\begin{Example}
Consider again the example of Section~\ref{sec:running-example}, this time in the setting of Game VI.
Consider also the calculations made in Example~\ref{exa:exempio}, we will use the same results and notation here. 
In this setting, the attacker is obliged to make its choice first. 
If he chooses $0$, which corresponds to committing to the system $C_{00}\hchoice{p} C_{10}$, 
then the defender will choose $p=\nicefrac{1}{4}$, which minimizes its vulnerability. 
If he chooses $1$, which corresponds to committing to the system $C_{01}\hchoice{p} C_{11}$, 
the defender will choose $p=1$, which minimizes its vulnerability of the above channel. 
In both cases, the leakage is $p=\nicefrac{1}{2}$, hence both these strategies are solutions to the minimax. 
Note that in the first case the strategy of the defender is mixed, while that of the attacker is always pure. 
\end{Example}

\subsection{Comparing the games}
\label{sec:comparing-games}

If we look at the various payoffs obtained for the running example in the various games, we obtain the following values (listed in decreasing order):
${\rm II:} \, 1;\, \allowbreak 
{\rm I:} \, \nicefrac{4}{5};\, \allowbreak 
{\rm III:} \, \nicefrac{2}{3};\, \allowbreak 
{\rm IV:} \, \nicefrac{4}{7};\, \allowbreak 
{\rm V:} \, \nicefrac{4}{7};\, \allowbreak 
{\rm VI:} \, \nicefrac{1}{2} \allowbreak 
$.
\commentM{Convert this into a theorem. 
Then probably Theorem~\ref{theo:order} will become a lemma.
And we have to check the relation of all this with Theorem~\ref{theo:deterministic-strategies}.}
\replyY{Let's do it in the journal version of this POST paper.}

\begin{wrapfigure}{r}{0.3\linewidth}
\centering
\includegraphics[width=0.14\columnwidth]{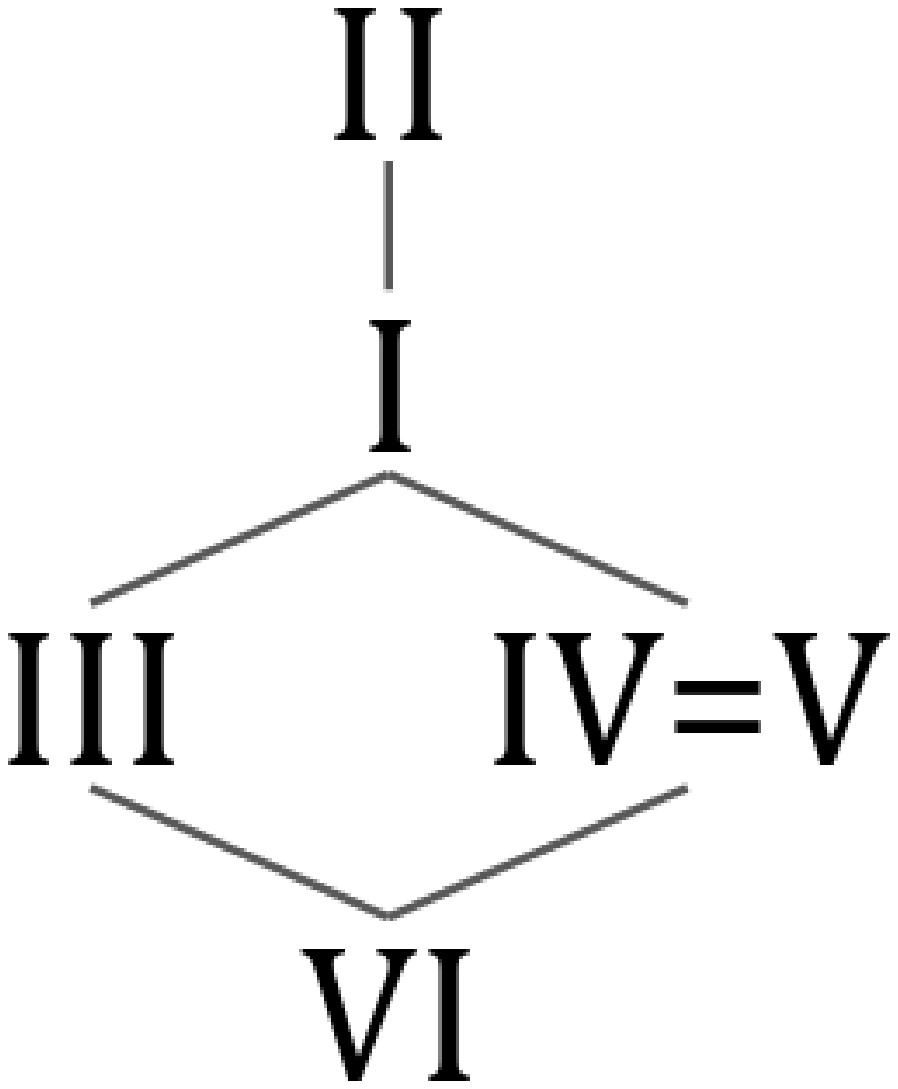}
\vspace{-2mm}
\caption{Order of games w.r.t. payoff. 
Games higher in the lattice have larger payoff.}
\vspace{-6mm}
\label{fig:order-games1}
\end{wrapfigure}
This order is not accidental:
for any vulnerability function, and for any prior, the various games are
ordered, with respect to the payoff, as shown in Figure~\ref{fig:order-games1}.
The relations between II, I, and III, and between IV-V and VI come from the fact that, in any zero-sum sequential game the leader's payoff  will
be less or equal to his payoff in the corresponding simultaneous game.
We think this result is well-known in  game theory, but we give the hint of the proof nevertheless, for the sake of clarity.
\begin{restatable}{Theorem}{order}
\label{theo:order}
It is the case that:
\begin{enumerate}[(a)]
\item\label{item:order-uno}
$
\begin{array}[t]{lll}
\min_\delta\max_\msa \postvf{\pi}{\VChoiceDouble{d}{\delta}{\psa}{\msa}  C_{d \psa(d)}} &\geq &
\min_\delta\max_\alpha \postvf{\pi}{\VChoiceDouble{d}{\delta}{a}{\alpha} C_{da}} \\
&\geq&
\max_\alpha  \min_{\msd}{\postvf{\pi}{\VChoiceDouble{\psd}{\msd}{a}{\alpha}  C_{\psd(a) a}} }
\end{array}$\\
\,\\
\item
$
\min_\delta\max_\alpha\postvf{\pi}{\VChoice{a}{\alpha}\, {\HChoice{d}{\delta}} {C_{d a}}} \geq
\max_\alpha  \min_{\msd} {\postvf{\pi}{\HChoice{a}{\alpha}\;\VChoice{\psd}{\msd}  C_{\psd(a) a}} }
$
\end{enumerate}
\end{restatable}
\begin{proof}
We prove the first inequality in ({\it\ref{item:order-uno}}). 
Independently of $\delta$, consider the attacker strategy $\tau_a$ that assigns probability $1$ to the function $\psa$ defined as
$
\psa(d) = \argmax_a \postvf{\pi}{C_{d a}} 
$.
Then we have that
\begin{align*}
\min_\delta\max_\msa \postvf{\pi}{\VChoiceDouble{d}{\delta}{\psa}{\msa}  C_{d \psa(d)}} 
\geq&\,\, \min_\delta \postvf{\pi}{\VChoiceDouble{d}{\delta}{\psa}{\tau_a}  C_{d \psa(d)}}  \\
\geq&\,\, \min_\delta\max_\alpha \postvf{\pi}{\VChoiceDouble{d}{\delta}{a}{\alpha} C_{da}} 
\end{align*}
Note that the strategy $\tau_a$ is optimal  for the adversary, so the first of the above inequalities is actually an equality. 
All other cases can be proved with an analogous reasoning. 
\qed
\end{proof}

Concerning III and IV-V: these are not related. In the running example the payoff for III is higher than for IV-V, but it is easy to find other cases in which the situation is reversed. For instance, if in the running example we set $C_{11}$ to be the same as $C_{00}$, the payoff for  III  will be $\nicefrac{1}{2}$, and that for IV-V will be $\nicefrac{2}{3}$.

Finally, the relation between III  and VI comes from the fact that they are both attacker-first sequential games, and 
the only difference is the way in which the payoff is defined. 
Then, just observe that in general we have, for every $a\in \cala$ and every $\delta\in \distr\cald$:
$
\postvf{\pi}{{\HChoice{d}{\delta}} {C_{d a}}}\leq\postvf{\pi}{{\VChoice{d}{\delta}} {C_{d a}}}
$.
 
The relations in Figure~\ref{fig:order-games1} can be used by the defender as guidelines to better protect the system, if he has some control over the rules of the game. Obviously, for the defender the games lower in the ordering are to be preferred.

\section{Case study: a safer, faster password-checker}
\label{sec:password-example}
\begin{wrapfigure}[14]{r}{0.44\linewidth}
\vspace{-8mm}
\begin{footnotesize}
\noindent \texttt{\underline{Program PWD\textsubscript{123}}}\\[1.3mm]
\texttt{\textbf{High Input:}} $x \in \{000, 001, \ldots, 111\}$\\[-0.3mm]
\texttt{\textbf{Low Input:}} $a \in \{000, 001, \ldots, 111\}$\\[-0.3mm]
\texttt{\textbf{Output:}} $y\in \{\true,\false\}$\\[0.3mm] 
accept $:= \true$ \\[-0.3mm]
\textbf{for} $i = 1, 2, 3$ \textbf{do}\\[-0.3mm]
\hspace*{4mm} \textbf{if} $a_{i} \neq x_{i}$ \textbf{then}\\[-0.3mm]
\hspace*{4mm} \hspace*{4mm} accept $:= \false$\\[-0.3mm]
\hspace*{4mm} \hspace*{4mm} \textbf{break}\\[-0.3mm]
\hspace*{4mm} \textbf{end if}\\[-0.3mm]
\textbf{end for}\\[-0.3mm]
\textbf{return} accept
\end{footnotesize}
\vspace{-2mm}
\caption{Password-checker algorithm.}
\label{fig:pwd-checker}
\end{wrapfigure}
In this section we apply our game-theoretic, compositional
approach to show how a defender can mitigate an attacker's typical timing 
side-channel attack while avoiding the usual burden imposed on the 
password-checker's efficiency.

Consider the password-checker \texttt{PWD}\textsubscript{123} of 
Figure~\ref{fig:pwd-checker}, which performs a bitwise-check of
a 3-bit low-input $a{=}a_{1}a_{2}a_{3}$, provided by the attacker,
against a 3-bit secret password $x{=}x_{1}x_{2}x_{3}$.
The low-input is rejected as soon as it mismatches the secret, 
and is accepted otherwise.

The attacker can choose low-inputs to try to gain information about 
the password. 
Obviously, in case \texttt{PWD}\textsubscript{123} accepts the low-input,
the attacker learns the password value is $a{=}x$.
Yet, even when the low-input is rejected,
there is some leakage of information:
from the duration of the execution the attacker can estimate how 
many iterations have been performed before the low-input was rejected, 
thus inferring a prefix of the secret password. 

To model this scenario, let $\calx = \{000,001,\ldots,111\}$ be the set of all possible 3-bit 
passwords, and $\caly = \{(\false,1), (\false,2), \allowbreak (\false,3), (\true,3)\}$ be 
the set of observables produced by the system.
Each observable is an ordered pair whose first element indicates whether the password 
was accepted ($\true$ or $\false$), and the second element indicates the duration of 
the computation ($1$, $2$, or $3$ iterations).
For instance, channel $C_{123,101}$ in Figure~\ref{fig:pwd-channels}
models \texttt{PWD}\textsubscript{123}'s behavior when the 
attacker provides low-input $a{=}101$.

We will adopt as a measure of information \emph{Bayes vulnerability}~\cite{Smith:09:FOSSACS}.
The \emph{prior Bayes vulnerability} of a distribution $\pi{\in}\dist{\calx}$ is 
defined as
$
\priorvg{\pi}{=}\max_{x \in \calx} \pi_x
$,
and represents the probability that the attacker guesses correctly
the password in one try.
For instance, if the distribution on all possible $3$-bit passwords is 
$\hat{\pi}=(0.0137, 0.0548, 0.2191, \allowbreak 0.4382, \allowbreak 0.0002, \allowbreak 0.0002, \allowbreak 0.0548, \allowbreak 0.2191)$, 
its prior Bayes vulnerability is $\priorvf{\hat{\pi}}{=}0.4382$.

The \emph{posterior Bayes vulnerability} of a prior $\pi$ and a channel 
$C{:}\calx{\times}\caly{\rightarrow}\reals$ is defined as
$
\postvf{\pi}{C}{=}\sum_{y \in \caly} \max_{x \in \calx} \pi_x C(x,y)
$,
and it represents the probability that the attacker guesses correctly
the password in one try, after he observes the output of the channel (i.e.,
after he has measured the time needed for the checker to accept or reject 
the low-input).
For prior $\hat{\pi}$ above, the posterior 
Bayes 
vulnerability of channel $C_{123,101}$ is $\postvf{\hat{\pi}}{C_{123,101}}{=}0.6577$
(which represents an 
increase in Bayes vulnerability 
of about $50\%$), and the expected running time for this checker is of $1.2747$ iterations.

A way to mitigate this timing side-channel is to make the checker's
execution time independent of the secret. 
Channel $C_{\text{cons},101}$ from Figure~\ref{fig:pwd-channels} models
a checker that does that (by eliminating the \texttt{break} command within 
the loop in \texttt{PWD}\textsubscript{123}) when the attacker's low-input 
is $a{=}101$.
This channel's posterior 
Bayes vulnerability is $\postvf{\hat{\pi}}{C_{123,101}}{=}0.4384$, 
which brings the multiplicative
Bayes leakage down to an increase of only about $0.05\%$.
However, the expected running time goes up to $3$ iterations 
(an increase of about $135\%$ w.r.t. that of $C_{123,101}$).

\begin{wrapfigure}{r}{0.61\linewidth}
\vspace{-8mm}
\centering
\begin{scriptsize}
$
\begin{array}{|c|c|c|c|c|}\hline
\multirow{2}{*}{$C_{123,101}$} & y{=} & y{=} & y{=} & y{=} \\
 & (\false,1) & (\false,2) & (\false,3) & (\true,3) \\ \hline
x{=}000    & 1 & 0 & 0 & 0 \\
x{=}001    & 1 & 0 & 0 & 0 \\
x{=}010    & 1 & 0 & 0 & 0 \\
x{=}011    & 1 & 0 & 0 & 0 \\
x{=}100    & 0 & 0 & 1 & 0 \\
x{=}101    & 0 & 0 & 0 & 1 \\
x{=}110    & 0 & 1 & 0 & 0 \\
x{=}111    & 0 & 1 & 0 & 0 \\ \hline
\end{array}
$
\,\,
$
\begin{array}{|c|c|c|}
\hline
\multirow{2}{*}{$C_{\text{cons},101}$} & y{=} & y{=} \\ 
& (\false,3) & (\true,3) \\ \hline
x{=}000    & 1 & 0  \\
x{=}001    & 1 & 0  \\
x{=}010    & 1 & 0  \\
x{=}011    & 1 & 0  \\
x{=}100    & 1 & 0  \\
x{=}101    & 0 & 1  \\
x{=}110    & 1 & 0  \\
x{=}111    & 1 & 0  \\ \hline
\end{array}
$
\end{scriptsize}
\caption{Channels $C_{da}$ modeling the password
checker for defender's action $d$ and attacker's action $a$.}
\label{fig:pwd-channels}
\vspace{-8mm}
\end{wrapfigure}
Seeking some compromise between security and efficiency, assume that 
the defender can employ password-checkers that perform the bitwise 
comparison among low-input $a$ and secret password $x$ in 
different orders.
More precisely, there is one version of the checker for every 
possible order in which the index $i$ 
ranges in the control of the loop.
For instance, while \texttt{PWD}\textsubscript{123} checks the bits in the order 
$1,2,3$, the alternative algorithm \texttt{PWD}\textsubscript{231} uses the order 
$2,3,1$.

To determine a defender's best choice of which versions of the checker to run,
we model this problem as game.
The attacker's actions $\cala = \{000, 001, \ldots, 111\}$ are all possible low-inputs
to the checker, 
and the defender's $\cald = \{123,  \allowbreak 132,  \allowbreak 213, \allowbreak  231,  \allowbreak 312, \allowbreak 321\}$ are all orders
to perform the comparison.
Hence, there is a total of $48$ possible channels 
$C_{ad}{:}\calx{\times}\caly{\rightarrow}\reals$, one for
each combination of $d{\in}\cald$, $a{\in}\cala$.

In our framework, the utility of a mixed strategy profile
$(\delta, \alpha)$ is given by
$\Pay(\delta,\alpha) \allowbreak = \allowbreak \expectDouble{a\leftarrow\alpha}{}\hspace{-1ex} \; \postvf{\pi}{{\HChoice{d}{\delta}} {C_{d a}}}$.
For each pure strategy profile $(d,a)$, the payoff of the game will be the posterior 
Bayes vulnerability
of the resulting channel $C_{da}$
(since, if we measuring leakage, the prior vulnerability is the same
for every channel once the prior is fixed).
Table~\ref{tab:utility} depicts such payoffs.
Note that the attacker's and defender's actions substantially affect
the effectiveness of the attack:
vulnerability ranges between 0.4934 and 0.9311 
(and so multiplicative leakage is in the range
between an increase of $12\%$ and one of $112\%$).
Using techniques from~\cite{Alvim:17:GameSec}, we can compute
the best (mixed) strategy for the defender in this game, which turns out 
to be
$
\delta^{*}  = (0.1667, 0.1667, 0.1667, 0.1667, 0.1667, 0.1667)
$.
This strategy is part of an equilibrium and guarantees that for any choice
of the attacker the posterior Bayes vulnerability is at most $0.6573$
(so the multiplicative leakage is bounded by $50\%$, an intermediate value
between the minimum of about $12\%$ and the maximum of about $112\%$).
It is interesting to note that the
expected running time, for any action
of the attacker, is bounded by at most
\begin{wraptable}{r}{0.71\linewidth}
\vspace{-8mm}
\centering
\begin{scriptsize}
$
\begin{array}{c|c||c|c|c|c|c|c|c|c|}
\multicolumn{2}{c}{} & \multicolumn{8}{c}{\textbf{Attacker's action $a$}} \\ \cline{2-10}
  & U(d,a) & 000 & 001 & 010 & 011 & 100 & 101 & 110 & 111 \\ \cline{2-10} \cline{2-10}
\multirow{6}{*}{\rotatebox[origin=c]{90}{\stackanchor{\textbf{Defender's}}{\textbf{action $d$}}}}
  & 123 &  0.7257 & 0.7257 & 0.9311 & 0.9311 & 0.6577 & 0.6577 & 0.7122 & 0.7122 \\ \cline{2-10}
  & 132 &  0.8900 & 0.9311 & 0.8900 & 0.9311 & 0.7122 & 0.7122 & 0.7122 & 0.7122 \\ \cline{2-10}
  & 213 &  0.5068 & 0.5068 & 0.9311 & 0.9311 & 0.4934 & 0.4934 & 0.7668 & 0.7668 \\ \cline{2-10}
  & 231 &  0.5068 & 0.5068 & 0.7668 & 0.9311 & 0.5068 & 0.5068 & 0.7668 & 0.9311 \\ \cline{2-10}
  & 312 &  0.7257 & 0.9311 & 0.7257 & 0.9311 & 0.7122 & 0.8766 & 0.7122 & 0.8766 \\ \cline{2-10}
  & 321 &  0.6712 & 0.7122 & 0.7257 & 0.9311 & 0.6712 & 0.7122 & 0.7257 & 0.9311 \\ \cline{2-10}
\end{array}
$
\end{scriptsize}
\vspace{-3mm}
\caption{Utility for each pure strategy profile.}
\label{tab:utility}
\vspace{-7mm}
\end{wraptable}
$2.3922$ iterations
(an increase of only $87\%$ w.r.t. the channel \texttt{PWD}$_{123}$),
which is below 
the worst possible expected $3$ iterations of 
the constant-time password checker.

\section{Related work}
\label{sec:related-work}
Many studies have applied game theory to analyses of security 
and privacy in networks~\cite{Basar:83:TIT,Grossklags:08:WWW,Alpcan:11:TMC},
cryptography~\cite{Katz:08:TCC},
anonymity~\cite{Acquisti:03:FC}, 
location privacy~\cite{Freudiger:09:CCS}, and
intrusion detection~\cite{Zhu:09:GAMENETS},
to cite a few.
See \cite{Manshaei:13:ACMCS} for a survey. 

In the context of quantitative information flow, 
most works consider only passive attackers.
Boreale and Pampaloni~\cite{Boreale:15:LMCS} consider adaptive
attackers, but not adaptive defenders, and show that 
in this case the adversary's optimal strategy can 
be always deterministic.
Mardziel et al.~\cite{Mardziel:14:SP} propose a model for both adaptive 
attackers and defenders, but in none of their extensive case-studies 
the attacker needs a probabilistic strategy to maximize leakage.
In this paper we characterize when randomization is necessary, 
for either attacker or defender, to achieve optimality in our general 
information leakage games.

Security games have been employed to model and analyze payoffs between 
interacting agents, especially between a defender and an attacker.
Korzhyk et al. \cite{Korzhyk:11:JAIR} theoretically analyze security 
games and study the relationships between Stackelberg and Nash Equilibria 
under various forms of imperfect information.
Khouzani and Malacaria \cite{Khouzani:16:CSF} study leakage properties when perfect
secrecy is not achievable due to constraints on the allowable size of the conflating
sets, and provide universally optimal strategies for a wide class of entropy measures,
and for $g$-entropies.
These works, contrarily to ours, do not consider games with hidden 
choice, in which optimal strategies differ from traditional game-theory.

Several security games have modeled leakage when the sensitive
information are the defender's choices themselves, rather than a system's 
high input.
For instance, Alon et al.~\cite{Alon:13:SIAMDM} propose 
zero-sum games in which a defender chooses probabilities of secrets and an attacker chooses and learns some of the defender's secrets.
Then they present how the leakage on the defender's secrets gives influences on the defender's optimal strategy.
More recently, Xu et al.~\cite{Xu:15:IJCAI} show zero-sum games in which the attacker 
obtains partial knowledge on the security resources that the defender protects, and provide 
the defender's optimal strategy under the attacker's such knowledge.

Regarding channel operators, sequential and parallel composition of channels
have been studied (e.g.,~\cite{Kawamoto:17:LMCS}), but we are unaware of any
explicit definition and investigation of hidden and visible choice operators.
Although Kawamoto et al.~\cite{KawamotoBL16} implicitly use the hidden choice to model 
a probabilistic system as the weighted sum of systems, they do not derive the set 
of algebraic properties we do for this operator, and for its interaction with the 
visible choice operator.

\section{Conclusion and future work}
\label{sec:conclusion}
In this paper we used protocol composition to model the introduction of noise 
performed by the defender to prevent leakage of sensitive information.
More precisely, we formalized visible and hidden probabilistic choices of different protocols.
We then formalized the interplay between defender and adversary in a game-theoretic framework adapted to the specific issues of QIF, where the payoff is information leakage. 
We considered various kinds of leakage games, depending on whether players act simultaneously or sequentially, and whether the choices of the defender are visible or not to the adversary.
We established a hierarchy of these games, and provided methods for finding the optimal strategies (at the points of equilibrium) in the various cases.

As future research, we would like to extend leakage games to the 
case of repeated observations, i.e., when the attacker can 
observe  the outcomes of the system in successive runs,  under the 
assumption that both attacker and defender may change the 
channel in each run. 
We would also like to extend our framework to non zero-sum games, in which
the costs of attack and defense are not equivalent, and to analyze differentially-private mechanisms.

\paragraph*{Acknowledgments}
\begin{small}
The authors are thankful to anonymous reviewers for helpful comments.
This work was supported by JSPS and Inria under the project LOGIS of the Japan-France AYAME Program,
and by the project Epistemic Interactive Concurrency (EPIC) from the STIC 
AmSud Program.
M\'{a}rio S. Alvim was supported by CNPq, CAPES, and FAPEMIG.
Yusuke Kawamoto was supported by JSPS KAKENHI Grant Number JP17K12667.
\end{small}

\bibliographystyle{plain}
\bibliography{short,new}

\appendix

\newpage
\section{Proofs of Technical Results}
\label{sec:proofs}

In this section we provide proofs for our technical results.

\subsection{Preliminaries for proofs}

We start by providing some necessary background for
the subsequent technical proofs.

Definition~\ref{def:equivalence-channels} states that
two compatible channels (i.e., with the same input space)
are equivalent if yield the same value of vulnerability 
for every prior and every vulnerability function.
The result below, from the literature, provides 
necessary and sufficient conditions for two channels being 
equivalent.
The result employs the extension of channels with an all-zero 
column as follows.
For any channel $C$ of type $\calx \times \caly \rightarrow \reals$, 
its \emph{zero-column extension} $C^{0}$ is the channel of type 
$\calx \times \left(\caly \cup \{y_{0}\}\right) \rightarrow \reals$,
with $y_{0} \notin \caly$, s.t. $C^{0}(x,y) = C(x,y)$ for all 
$x\in\calx$, $y\in\caly$, and $C^{0}(x,y_{0}) = 0$ for all $x \in \calx$.

\begin{restatable}[Characterization of channel equivalence~\cite{Alvim:12:CSF,McIver:14:POST}]{Lemma}{resequivalentchannelsconditions}
\label{lemma:equivalent-channels-conditions}
Two channels $C_{1}$ and $C_{2}$ are equivalent iff
every column of $C_{1}$ is a convex combination of 
columns of $C_{2}^{0}$, and every column of $C_{2}$ is a 
convex combination of columns of $C_{1}^{0}$.
\end{restatable}

Note that the result above implies that for being equivalent,
any two channels must be compatible.

\subsection{Proofs of Section~\ref{sec:operators}}

\restypehiddenchoice*

\begin{proof}
Since hidden choice is defined as a summation
of matrices, the type of $\HChoice{i}{\mu} C_{i}$ is 
the same as the type of every $C_{i}$ in the family.

To see that $\HChoice{i}{\mu} C_{i}$ is a channel
(i.e., all of its entries are non-negative, and all
of its rows sum up to 1),
first note that, since each $C_{i}$ in the family is a 
channel matrix, $C_{i}(x,y)$ lies in the interval 
$[0,1]$ for all $x \in \calx$, $y\in \caly$.
Since $\mu$ is a set of convex coefficients, from 
the definition of hidden choice
it follows that
also $\sum_{i \in \cali} \mu(i)C_{i}(x,y)$ must lie 
in the interval $[0,1]$ for every $x,y$.

Second, note that for all $x \in \calx$:
\begin{align*}
\sum_{y\in \caly} \left(\HChoice{i}{\mu} C_{i}\right)(x,y)
=&\, \sum_{y\in \caly} \sum_{i \in \cali} \mu(i)C_{i}(x,y) & \text{(def. of hidden choice)} \\
=&\, \sum_{i \in \cali} \mu(i) \sum_{y\in \caly} C_{i}(x,y) \\
=&\, \sum_{i \in \cali} \mu(i) \cdot 1 & \text{($C_{i}{:}\calx{\times}\caly{\rightarrow}\reals$ are channels)} \\
=&\, 1 & \text{($\mu$ is a prob. dist.)}
\end{align*}
\end{proof}

\restypevisiblechoice*

\begin{proof}
Visible choice applied to a family $\{C_{i}\}$ of channels 
scales each matrix $C_{i}$ by a factor $\mu(i)$, which
preserves the type $\calx \times \caly_{i} \rightarrow \reals$ 
of each matrix, and then concatenates all the matrices so produced,
yielding a result of type
$\calx \times \left( \bigsqcup_{i \in \cali} \caly_{i} \right) \rightarrow \reals$.

To see that $\VChoice{i}{\mu} C_{i}$ is a channel
(i.e., that all of its entries are non-negative, and that
all rows sum-up to $1$), note that each element of the 
visible choice on the family $\{C_{i}\}$ can be denoted by 
$\left(\VChoice{i}{\mu} C_{i}\right)(x,(y,j))$,
where $x \in \calx$, $j \in \cali$, and $y \in \caly_{j}$.
Then, note that for all $x \in \calx$, $j \in \cali$, and $y \in \caly_{j}$:
\begin{align}
\label{eq:propvchoicetype0}
\left(\VChoice{i}{\mu} C_{i}\right)(x,(y,j)) \nonumber 
=&\, \left(\bigconc_{i \in \cali} \mu(i)C_{i}\right)(x,(y,j)) & \text{(def. of visible choice)} \nonumber \\
=&\, \left(\mu(j)C_{j}\right)(x,y) & \text{(def. of concatenation)} \nonumber \\
=&\, \mu(j)C_{j}(x,y) & \text{(def. of scalar mult.)}
\end{align}
which is a non-negative value, since, both $\mu(j)$ and 
$C_{j}(x,y)$ are non-negative.

Finally, note that for all $x \in \calx$:
\begin{align*}
\sum_{\substack{j \in \cali \\ y \in \caly_{j}}} \left(\VChoice{i}{\mu} C_{i}\right)(x,(y,j))
=&\, \sum_{\substack{j \in \cali \\ y \in \caly_{j}}} \mu(j)C_{j}(x,y) &
\text{(by Eq.~\eqref{eq:propvchoicetype0})} \\ 
=&\, \sum_{j \in \cali} \mu(j) \sum_{y \in \caly_{j}} C_{j}(x,y) &
\text{} \\ 
=&\, \sum_{j \in \cali} \mu(j) \cdot 1 & \text{($C_{j}{:}\calx{\times}\caly_{j}{\rightarrow}\reals$ are channels)} \\
=&\, 1 & \text{($\mu$ is a prob. dist.)}
\end{align*}
\end{proof}

\residempotency*

\begin{proof}

\begin{enumerate}[a)]
\item Idempotency of hidden choice:
\begin{align*}
 \HChoice{i}{\mu}C_{i} 
=&\, \sum_{i} \mu(i)C_{i} & \text{(def. of hidden choice)} \\
=&\, \sum_{i} \mu(i)C & \text{(since every $C_{i}=C$)} \\
=&\, C \sum_{i} \mu(i) & \text{} \\
=&\, C & \text{(since $\mu$ is a prob. dist.)}
\end{align*}

\item Idempotency of visible choice:
\begin{align*}
 \VChoice{i}{\mu}C_{i} 
=&\, \bigconc_{i} \mu(i)C_{i} & \text{(def. of visible choice)} \\
=&\, \bigconc_{i} \mu(i)C & \text{(since every $C_{i}=C$)} \\
\equiv&\, C & \text{(by Lemma~\ref{lemma:equivalent-channels-conditions})}
\end{align*}

In the above derivation, we can apply 
Lemma~\ref{lemma:equivalent-channels-conditions}
because every column of the channel on each side of the equivalence 
can be written as a convex combination of the zero-column extension
of the channel on the other side.
\end{enumerate}
\end{proof}

\resreorganizationoperators*

\begin{proof}

\begin{enumerate}[a)]
\item 
\begin{align*}
\HChoice{i}{\mu} \HChoice{j}{\eta} C_{i j} 
=&\, \HChoice{i}{\mu} \left( \sum_{j} \eta(j) C_{i j} \right) & \text{(def. of hidden choice)} \\ 
=&\, \sum_{i} \mu(i) \left( \sum_{j} \eta(j) C_{i j} \right) & \text{(def. of hidden choice)} \\ 
=&\, \sum_{i,j} \eta(i)\eta(j) C_{i j} & \text{(reorganizing the sums)} \\ 
=&\, \HChoiceDouble{i}{\mu}{j}{\eta} C_{i j} & \text{(def. of hidden choice)}
\end{align*}

\item 
\begin{align*}
\VChoice{i}{\mu} \VChoice{j}{\eta} C_{i j} 
=&\, \VChoice{i}{\mu} \left( \bigconc_{j} C_{i j} \right) & \text{(def. of visible choice)} \\ 
=&\, \bigconc_{i} \mu(i) \left( \bigconc_{j} \eta(j) C_{i j} \right) & \text{(def. of visible choice)} \\
\equiv& \bigconc_{i j} \mu(i) \eta(j) C_{i j} & \text{(by Lemma~\ref{lemma:equivalent-channels-conditions})} \\ 
=&\,\VChoiceDouble{i}{\mu}{j}{\eta} C_{i j} & \text{(def. of visible choice)}
\end{align*}
In the above derivation, we can apply 
Lemma~\ref{lemma:equivalent-channels-conditions}
because every column of the channel on each side of the equivalence 
can be written as a convex combination of the zero-column extension
of the channel on the other side.

\item 
\begin{align*}
\HChoice{i}{\mu} \VChoice{j}{\eta} C_{i j} 
=&\, \HChoice{i}{\mu} \left( \bigconc_{j} C_{i j} \right) & \text{(def. of visible choice)} \\ 
=&\, \sum_{i} \mu(i) \left( \bigconc_{j} \eta(j) C_{i j} \right) & \text{(def. of hidden choice)} \\ 
\equiv&\,\bigconc_{j} \eta(j) \left( \sum_{i} \eta(i) C_{i j} \right) & \text{(by Lemma~\ref{lemma:equivalent-channels-conditions})} \\ 
=&\, \VChoice{j}{\eta} \HChoice{i}{\mu} C_{i j} & \text{(def. of operators)}
\end{align*}
In the above derivation, we can apply 
Lemma~\ref{lemma:equivalent-channels-conditions}
because every column of the channel on each side of the equivalence 
can be written as a convex combination of the zero-column extension
of the channel on the other side.
\end{enumerate}

\end{proof}

\resconvexVq*

\begin{proof}
\begin{enumerate}[a)]
\item Let us call $\calx \times \caly \rightarrow \reals$ the 
type of each channel $C_{i}$ in the family $\{C_{i}\}$.
Then:
\begin{align*}
\postvf{\pi}{\HChoice{i}{\mu} C_{i}} 
=&\, \postvf{\pi}{\bigadd_{i} \mu(i)C_{i}} & \text{(def. of hidden choice)} \\ 
=&\, \sum_{y \in \caly} p(y) \cdot \priorvf{\frac{\pi(\cdot) \sum_{i}\mu(i)C_{i}(\cdot,y) }{p(y)}}  & \text{(def. of posterior $\vf$)}\\ 
=&\, \sum_{y \in \caly} p(y) \cdot \priorvf{\sum_{i}\mu(i)\frac{\pi(\cdot) C_{i}(\cdot,y) }{p(y)}}  & \text{}\\ 
\leq&\, \sum_{y \in \caly} p(y) \cdot \sum_{i} \mu(i) \priorvf{\frac{\pi(\cdot) C_{i}(\cdot,y) }{p(y)}}  & \text{(by convexity of $\vf$)}\\ 
=&\, \sum_{i} \mu(i) \sum_{y \in \caly} p(y) \priorvf{\frac{\pi(\cdot) C_{i}(\cdot,y) }{p(y)}}  & \text{}\\
=&\, \sum_{i} \mu(i) \postvf{\pi}{C_{i}}  & \text{}
\end{align*}
where $p(y) = \sum_{x\in\calx} \pi(x) \sum_{i} \mu(i)C_{i}(x,y)$.

\item Let us call $\calx \times \caly_{i} \rightarrow \reals$ the 
type of each channel $C_{i}$ in the family $\{C_{i}\}$.
Then:
\begin{align*}
\postvf{\pi}{\VChoice{i}{\mu} C_{i}}
=&\, \postvf{\pi}{\bigconc_{i} \mu(i)C_{i}} & \text{(def. of visible choice)} \\ 
=&\, \sum_{y \in \caly} p(y) \cdot \priorvf{\frac{\pi(\cdot) \bigconc_{i}\mu(i)C_{i}(\cdot,y) }{p(y)}}  & \text{(def. of posterior $\vf$)}\\ 
=&\, \sum_{y \in \caly} p(y) \cdot \priorvf{\bigconc_{i}\mu(i)\frac{\pi(\cdot) C_{i}(\cdot,y) }{p(y)}}  & \text{}\\ 
=&\, \sum_{y \in \caly} p(y) \cdot \sum_{i} \mu(i) \priorvf{\frac{\pi(\cdot) C_{i}(\cdot,y) }{p(y)}}  & \text{(see (*) below)}\\ 
=&\, \sum_{i} \mu(i) \sum_{y \in \caly} p(y) \priorvf{\frac{\pi(\cdot) C_{i}(\cdot,y) }{p(y)}}  & \text{}\\
=&\, \sum_{i} \mu(i) \postvf{\pi}{C_{i}}  & \text{}
\end{align*}
where $p(y) = \sum_{x\in\calx} \pi(x) \sum_{i} \mu(i)C_{i}(x,y)$, and 
step (*) holds because in the vulnerability of a concatenation of matrices
every column will contribute to the vulnerability in proportion to its weight
in the concatenation, and hence it is possible to break the vulnerability of a 
concatenated matrix as the weighted sum of the vulnerabilities of its sub-matrices.

\end{enumerate}
\end{proof}

\resconcaveconvexV*

\begin{proof}
To see that $U(\mu,\eta)$ is convex on $\mu$, note that:
\begin{align*}
U(\mu,\eta) 
=&\, \postvf{\pi}{\HChoice{i}{\mu} \; \VChoice{j}{\eta} \; C_{i j}} & \text{(by definition)} \\
\leq&\, \sum_{i} \mu(i) \; \postvf{\pi}{\VChoice{j}{\eta} \; C_{i j}} & \text{(by Theorem~\ref{theo:convex-V-q})}
\end{align*}

To see that $U(\mu,\eta)$ is linear on $\eta$, note that:
\begin{align*}
U(\mu,\eta) 
=&\, \postvf{\pi}{\HChoice{i}{\mu} \; \VChoice{j}{\eta} \; C_{i j}} & \text{(by definition)} \\
=&\, \postvf{\pi}{\VChoice{j}{\eta} \; \HChoice{i}{\mu} \; C_{i j}} & \text{(by Prop.~\ref{prop:reorganization-operators})} \\
=&\, \sum_{j} \eta(j) \; \postvf{\pi}{\HChoice{i}{\mu} \; C_{i j}} & \text{(by Theorem.~\ref{theo:convex-V-q})}
\end{align*}

\end{proof}

\newpage
\section{Properties of binary versions of channel operators}
\label{sec:operators-properties-binary}
In this section we derive some relevant properties of 
the binary versions of the hidden and visible choice 
operators.
We start with results regarding each operator
individually.

\begin{restatable}[Properties of the binary hidden choice]{Proposition}{respropertiesbinaryhiddenchoice}
\label{prop:properties-binary-hidden-choice}
For any channels $C_{1}$ and $C_{2}$ of the same type, 
and any values $0 \leq p,q \leq 1$, 
the binary hidden choice operator satisfies the following 
properties:
\begin{enumerate}[a)]
\item \emph{Idempotency}:
$C_{1} \hchoice{p} C_{1} = C_{1}$
\item \emph{Commutativity}: 
$C_{1} \hchoice{p} C_{2} = C_{2} \hchoice{\bar{p}} C_{1}$
\item \emph{Associativity}: 
$$C_{1} \hchoice{p} (C_{2} \hchoice{q} C_{3}) = ( \nicefrac{1}{q} \cdot C_{1} \hchoice{p} C_{2}) \hchoice{q} \bar{p} \cdot C_{3}$$ 
if $q \neq 0$.
\item \emph{Absorption}: 
$$(C_{1} \hchoice{p} C_{2}) \hchoice{q} (C_{1} \hchoice{r} C_{2})
= C_{1} \hchoice{(pq + \bar{q}r)} C_{2}.$$
\end{enumerate}
\end{restatable}

\begin{proof}
We will prove each property separately.

\begin{enumerate}[a)]
\item \emph{Idempotency}: 
\begin{align*}
C_{1} \hchoice{p} C_{1}
=&\, p \cdot C_{1} + \bar{p} \cdot C_{1} & \text{(def. of hidden choice)} \\
=&\, (p + \bar{p}) \cdot C_{1} & \text{} \\
=&\, C_{1} & \text{($p+\bar{p}=1$)}
\end{align*}

\item \emph{Commutativity}:
\begin{align*}
C_{1} \hchoice{p} C_{2}
=&\, p \cdot C_{1} + \bar{p} \cdot C_{2} & \text{(def. of hidden choice)} \\
=&\, \bar{p} \cdot C_{2} + p \cdot C_{1} & \text{}\\
=&\, C_{2} \hchoice{\bar{p}} C_{1} & \text{(def. of hidden choice)}
\end{align*}

\item \emph{Associativity}:
\begin{align*}
 &\, C_{1} \hchoice{p} ( C_{2} \hchoice{q} C_{3} ) \\
=&\, p \cdot C_{1} + \bar{p} ( q \cdot C_{2} + \bar{q} \cdot C_{3} ) & \text{(def. of hidden choice)} \\
=&\,p \cdot C_{1} + \bar{p}q \cdot C_{2} + \bar{p}\bar{q} \cdot C_{3} & \text{} \\
=&\,q(p (\nicefrac{1}{q} \cdot C_{1}) + \bar{p} \cdot C_{2}) + \bar{q}(\bar{p} \cdot C_{3}) & \text{} \\
=&\,(\nicefrac{1}{q} \cdot C_{1} \hchoice{p} C_{2}) \hchoice{q} \bar{p} \cdot C_{3} & \text{(def. of hidden choice)} \\
\end{align*}

\item \emph{Absorption}: First note that:
\begin{align*}
 &\,(C_{1} \hchoice{p} C_{2}) \hchoice{q} (C_{1} \hchoice{r} C_{2}) \\
=&\,q (p C_{1} + \bar{p} C_{2}) + \bar{q} (r C_{1} + \bar{r} C_{2}) & \text{(def. of hidden choice)} \\
=&\,pq C_{1} + \bar{p}q C_{2} + \bar{q}r C_{1} + \bar{q}\bar{r} C_{2} \\
=&\,(pq + \bar{q}r) C_{1} + (\bar{p}q + \bar{q}\bar{r}) C_{2} \\
=&\,C_{1} \hchoice{(pq + \bar{q}r)} C_{2} & \text{(*)}
\end{align*}
To complete the proof, note that in step (*) above,
$pq + \bar{q}r$ and $\bar{p}q + \bar{q}\bar{r}$ form
a valid binary probability distribution (they are both 
non-negative, and they add up to 1), then apply the
definition of hidden choice.

\end{enumerate}
\end{proof}

\begin{restatable}[Properties of binary visible choice]{Proposition}{respropertiesbinaryvisiblechoice}
\label{prop:properties-binary-visible-choice}
For any compatible channels $C_{1}$ and $C_{2}$, 
and any values $0 \leq p,q \leq 1$, 
the visible choice operator satisfies the following 
properties:
\begin{enumerate}[a)]
\item \emph{Idempotency}: 
$C_{1} \vchoice{p} C_{1} \equiv C_{1}$
\item \emph{Commutativity}: 
$C_{1} \vchoice{p} C_{2} \equiv C_{2} \vchoice{\bar{p}} C_{1}$
\item \emph{Associativity}: 
$C_{1} \vchoice{p} (C_{2} \vchoice{q} C_{3}) \equiv ( \nicefrac{1}{q} \cdot C_{1} \vchoice{p} C_{2}) 
\vchoice{q} \bar{p} \cdot C_{3}$
if $q \neq 0$.
\end{enumerate}
\end{restatable}

\begin{proof}

We will prove each property separately.

\begin{enumerate}[a)]
\item \emph{Idempotency}: $C_{1} \vchoice{p} C_{1} \equiv C_{1}$, by 
immediate application of Lemma~\ref{lemma:equivalent-channels-conditions},
since every column of the channel on each side of the equivalence 
can be written as a convex combination of the zero-column extension
of the channel on the other side.

\item \emph{Commutativity}:
\begin{align*}
C_{1} \vchoice{p} C_{2} 
=&\, p \cdot C_{1} \conc \bar{p} \cdot C_{2} & \text{(def. of visible choice)} \\
\equiv&\ \bar{p} \cdot C_{2} \conc p \cdot C_{1} & \text{(by Lemma~\ref{lemma:equivalent-channels-conditions})} \\
=&\ C_{2} \vchoice{\bar{p}} C_{1} & \text{(def. of visible choice).} \\
\end{align*}
In the above derivation, we can apply 
Lemma~\ref{lemma:equivalent-channels-conditions}
because every column of the channel on each side of the equivalence 
can be written as a convex combination of the zero-column extension
of the channel on the other side.

\item \emph{Associativity}:
\begin{align*}
C_{1} \vchoice{p} ( C_{2} \vchoice{q} C_{3} ) 
=&\, p \cdot C_{1} \conc \bar{p} ( q \cdot C_{2} \conc \bar{q} \cdot C_{3} ) & \text{(def. of visible choice)} \\
\equiv&\, p \cdot C_{1} \conc \left( \bar{p}q \cdot C_{2} \conc \bar{p}\bar{q} \cdot C_{3} \right) & \text{(by Lemma~\ref{lemma:equivalent-channels-conditions})} \\
=&\, q(p (\nicefrac{1}{q} \cdot C_{1}) \conc \bar{p} \cdot C_{2}) \conc \bar{q}(\bar{p} \cdot C_{3}) & \text{} \\
=&\,(\nicefrac{1}{q} \cdot C_{1} \vchoice{p} C_{2}) \vchoice{q} \bar{p} \cdot C_{3} & \text{(def. of visible choice)}
\end{align*}
In the above derivation, we can apply 
Lemma~\ref{lemma:equivalent-channels-conditions}
because every column of the channel on each side of the equivalence 
can be written as a convex combination of the zero-column extension
of the channel on the other side.

\end{enumerate}
\end{proof}

Now we turn our attention to the interaction between
hidden and visible choice operators.

A first result is that hidden choice does not distribute 
over visible choice.
To see why, note that $C_{1} \hchoice{p} (C_{2} \vchoice{q} C_{3})$
and $(C_{1} \hchoice{p} C_{2}) \vchoice{q} (C_{1} \hchoice{p} C_{3})$
cannot be both defined: if the former is defined, then $C_{1}$ must have
the same type as $C_{2} \vchoice{q} C_{3}$, whereas if the
latter is defined, $C_{1}$ must have the same type as $C_{2}$, but
$C_{2} \vchoice{q} C_{3}$ and $C_{2}$ do not have the same type
(they have different output sets).

However, as the next result shows, visible choice distributes over 
hidden choice.

\begin{restatable}[Distributivity of $\vchoice{p}$ over $\hchoice{q}$]{Proposition}{resdistributivity}
\label{prop:distributivity}
Let $C_{1}$, $C_{2}$ and $C_{3}$ be compatible channels,
and $C_{2}$ and $C_{3}$ have the same type.
Then, for any values $0 \leq p,q \leq 1$:
\begin{align*}
C_{1} \vchoice{p} (C_{2} \hchoice{q} C_{3}) \equiv (C_{1} \vchoice{p} C_{2}) \hchoice{q} (C_{1} \vchoice{p} C_{3}).
\end{align*}
\end{restatable}

\begin{proof}
\begin{align*}
 &\, C_{1} \vchoice{p} (C_{2} \hchoice{q} C_{3}) \\
=&\, (C_{1} \hchoice{q} C_{1}) \vchoice{p} (C_{2} \hchoice{q} C_{3}) & \text{(idempotency of visible choice)} \\
=&\, p (q \cdot C_{1} + \bar{q} \cdot C_{1}) \conc \bar{p}(q \cdot C_{2} + \bar{q} \cdot C_{3}) & \text{(def. of operators)} \\
=&\, (pq \cdot C_{1} + p\bar{q} \cdot C_{1}) \conc (\bar{p}q \cdot C_{2} + \bar{p}\bar{q} \cdot C_{3}) & \text{} \\
\equiv&\,  (pq \cdot C_{1} \conc \bar{p}q \cdot C_{2}) + (p\bar{q} \cdot C_{1} \conc \bar{p}\bar{q} \cdot C_{3}) & \text{(by Lemma~\ref{lemma:equivalent-channels-conditions})} \\
=&\, q(p \cdot C_{1} \conc \bar{p} \cdot C_{2}) + \bar{q}(p \cdot C_{1} \conc \bar{p} \cdot C_{3}) & \text{} \\
=&\, q(C_{1} \vchoice{p} C_{2}) + \bar{q}(C_{1} \vchoice{p} C_{3})  & \text{(def. of visible choice)} \\
=&\, (C_{1} \vchoice{p} C_{2}) \hchoice{q} (C_{1} \vchoice{p} C_{3})  & \text{(def. of hidden choice)}
\end{align*}
In the above derivation, we can apply 
Lemma~\ref{lemma:equivalent-channels-conditions}
because every column of the channel on each side of the equivalence 
can be written as a convex combination of the zero-column extension
of the channel on the other side.
\end{proof}

\end{document}